\documentclass[reqno,12pt]{amsart}
\usepackage{amssymb}
\usepackage{amsxtra}
\usepackage{amsmath}
\usepackage{amsfonts}
\usepackage{appendix}
\usepackage{mathrsfs}

\newcommand{\bm}[1]{\mbox{\boldmath{$#1$}}}

\numberwithin{equation}{section} \allowdisplaybreaks
\newtheorem{thm}{Theorem}[section]
\newtheorem*{Mthm}{Main Theorem}
\newtheorem{prop}[thm]{Proposition}
\newtheorem{cor}[thm]{Corollary}

\newtheorem{lem}[thm]{Lemma}
\newtheorem{de}[thm]{Definition}
\newtheorem{rem}[thm]{Remark}

\newcommand{\eqa}{\begin{eqnarray}}
\newcommand{\eeqa}{\end{eqnarray}}
\newcommand{\beq}{\begin{equation}}
\newcommand{\eeq}{\end{equation}}
\newcommand{\nn}{\nonumber}
\newcommand{\p}{\partial}

\newcommand{\ve}{\epsilon}

\newcommand{\lm}{\lambda}

\newcommand{\al}{\alpha}

\newcommand{\sg}{\sigma}

\newcommand{\ta}{\theta}

\newcommand{\res}{\mathrm{res}}

\newcommand{\cM}{\mathcal{M}}

\newcommand{\mathZ}{\mathbb{Z}}

\newcommand{\mathH}{\mathcal{H}}

\newcommand{\mathL}{\mathcal{L}}
\newcommand{\mathM}{\mathcal{M}}

\newcommand{\mathT}{{T}}

\def \om{\omega}
\def \ep{\epsilon}
\def \la {\langle}
\def \ra{\rangle}
\def \dsum{\displaystyle\sum}

\newenvironment{prfof}[1]{\noindent {\it Proof of #1} \ }{\hfill $\Box$}

\begin{document}

\title[]
{\small Infinite-dimensional Frobenius Manifolds Underlying the Toda
Lattice Hierarchy}

\author[]{Chao-Zhong Wu and Dafeng Zuo}

\address {Wu, Marie Curie fellow of the Istituto
Nazionale di Alta Matematica,
SISSA, Via Bonomea 265, 34136 Trieste, Italy}
 \email{wucz@sissa.it}

\address{Zuo, School of Mathematical Science,
University of Science and Technology of China,
 Hefei 230026,
P.R.China and Wu Wen-Tsun Key Laboratory of Mathematics, USTC,
Chinese Academy of Sciences}

\email{dfzuo@ustc.edu.cn}



\date{\today}
\begin{abstract}
Following the approach of Carlet et al.(2011)\cite{CDM}, 
we construct a class of infinite-dimensional Frobenius
manifolds underlying the Toda lattice hierarchy, which are defined
on the space of pairs of meromorphic functions with possibly
higher-order poles at the origin and at infinity. We also show a
connection between these infinite-dimensional Frobenius manifolds and
the finite-dimensional Frobenius manifolds on the orbit space
of extended affine Weyl groups of type $A$ defined by Dubrovin and Zhang.
\end{abstract}

\maketitle 

{\small \tableofcontents}
\section{Introduction}

The concept of Frobenius manifold was introduced by Dubrovin
\cite{Du92, Du} as a geometric formalism of the
Witten-Dijkgraaf-E.\,Verlinde-H.\,Verlinde (WDVV) equation in
topological field theory  \cite{DVV, Witten}. Associated to every
Frobenius manifold, there is a so-called principal hierarchy of
Hamiltonian equations of hydrodynamic type, in which the unknown
functions depend on one scalar spatial variable and some time
variables. Under certain general assumptions, the principal
hierarchy can be extended to a full hierarchy \cite{DZ, DLZ} with a
tau function that plays an important role in relevant research areas
such like topological field theory.

Trials to extend the above programme to integrable systems with two
spatial dimensions started in recent years. The first step was made
by Carlet, Dubrovin and Mertens \cite{CDM} in consideration of the
dispersionless (2\,D) Toda lattice hierarchy \cite{UT}, in which the
first nontrivial $2+1$ evolutionary equation reads
\begin{equation}\label{utxy}
\p_t^2 u=\p_x^2 e^u+\p_y^2 u \nn
\end{equation}
for unknown function $u=u(x,y,t)$. It turns out that the underlying
Frobenius manifold is of infinite dimension. Following their
approach, a class of infinite-dimensional Frobenius manifolds was
constructed by Xu and one of the authors \cite{WX2} for the
dispersionless two-component BKP hierarchy, and these manifolds are
related to finite-dimensional Frobenius manifolds corresponding to
Coxeter groups of types B and D \cite{Ber3, Du, Z2007}. Along the
same line, now we want to revise and generalize the construction in
\cite{CDM} with some skills developed in \cite{WX2}. More exactly,
for the dispersionless Toda lattice hierarchy, we will find a family
of infinite-dimensional Frobenius manifolds labeled by a pair of
positive integers $(M,N)$, among which the $(1,1)$ case is similar to,
but not exactly the same as, with the Frobenius manifold given in
\cite{CDM}. Moreover, these infinite-dimensional manifolds will be
reduced to Frobenius manifolds of $M+N$ dimension that were
constructed on the orbit space of extended affine Weyl group
$\tilde{W}^{(N)}(A_{M+N-1})$ \cite{DZ2}. Thus we have more examples
to support a normal form of infinite-dimensional Frobenius manifolds
and their connection to Frobenius manifolds of finite dimension.

We remark that, with a method very different from that mentioned
above, Raimondo proposed an infinite-dimensional Frobenius manifold
for the dispersionless Kadomtsev-Petviashvili (KP) equation by using
the theory of Schwartz functions. In a recent paper \cite{Sz},
Szablikowski proposed a scheme for constructing Frobenius manifolds
of finite or infinite dimension based on the Rota-Baxter identity
and a counterpart of the modified Yang-Baxter equation for classical
$r$-matrix (see references therein). He suggested an
infinite-dimensional Frobenius manifold structure on the space of
Laurent series where the dispersionless KP hierarchy is defined, as
well as hints to study the case of space consisting of two component
of Laurent series to define the dispersionless Toda lattice
hierarchy. The relation between the methods in \cite{Ra, Sz} and the
approaches mentioned previously to infinite-dimensional Frobenius
manifolds is not clear yet.

Recall that a \emph{Frobenius algebra} $(A, e, <\ ,\ >)$  is a
   commutative associative algebra  $A$ with a unity $e$ and
   a non-degenerate symmetric  invariant bilinear form (inner product) $<\ ,\
   >$. A manifold $M$ is called a \emph{Frobenius
manifold} if
   on each tangent space $T_t M$ a Frobenius algebra $(T_t M, e, <\ ,\ >)$ is defined
     depending smoothly on $t\in M$, and the following conditions
     are satisfied:
  \begin{itemize}
     \item[(F1)] the inner product $<\ ,\ >$ is a flat metric on
     $M$, and, with $\nabla$ being the Levi-Civita
     connection for this metric, the
     unity vector field $e$ satisfies $\nabla e=0$;
     \item[(F2)] let $c$ be the 3-tensor $c(X,Y,Z):=<X\cdot Y,Z>$, then the
     4-tensor $(\nabla_{W}c)(X,Y,Z)$ is symmetric in the vector fields $X$, $Y$, $Z$ and $W$;
     \item[(F3)] there is a so-called Euler vector field $E$ on $M$
     such that $\nabla\nabla E=0$, and
     \begin{align}
        [E,X\cdot Y]-[E,X]\cdot Y-X\cdot [E,Y]&=X\cdot Y,\nn\\
        Lie_E <X,Y>-<[E,X],Y>-<X,[E,Y]>&=(2-d)<X,Y>, \nn
    \end{align}
    where $d$ is a constant named as the charge of $M$.
  \end{itemize}

On an $n$-dimensional Frobenius manifold $M$, one can choose a
system of flat coordinates $t=(t^1,\dots,t^n)$ such that the unity
vector field is $e=\p/\p t^1$. By using these coordinates, one has a
constant non-degenerate $n\times n$ matrix given by
   \begin{align}
     \eta_{\al\beta}=<\frac{\p}{\p t^\al},\frac{\p}{\p
     t^{\beta}}>,
  \nn  \end{align}
and its inverse denoted by $(\eta^{\al\beta})$.
The matrices $\eta_{\al\beta}$ and $\eta^{\al\beta}$ will be used to
lower and to lift indexes, respectively. Let
\begin{equation}\label{}
c_{\al\beta\gamma}=c\left(\frac{\p}{\p t^\al},\frac{\p}{\p t^\beta},
\frac{\p}{\p t^\gamma}\right),
\nn \end{equation}
then the product of the Frobenius algebra   $T_t M$ reads
   \begin{align}
     \frac{\p}{\p t^{\al}}\cdot \frac{\p}{\p t^{\beta}}=c^{\gamma}_{\al\beta}\frac{\p}{\p
     t^{\gamma}}, \quad
     c^{\gamma}_{\al\beta}=\eta^{\gamma\ep}c_{\ep\al\beta}.
  \nn \end{align}
Here and below the convention of summation over repeated Greek
indices is assumed.

The structure constants of the Frobenius algebra  $T_t M$ satisfy
\begin{equation}\label{WDVV1}
c_{1\al}^{\beta}=\delta_\al^\beta, \quad c_{\al\beta}^\ep
c_{\ep\gamma}^{\sg}=c_{\al\gamma}^\ep c_{\ep\beta}^{\sg}.
\end{equation}
Moreover, there locally exists a smooth function $F(t)$, named as
\emph{potential} of the Frobenius manifold, such that
\begin{align}
&c_{\al\beta\gamma}=\frac{\p^3 F}{\p t^{\al}\p t^{\beta}\p
t^{\gamma}}, \\
& Lie_E F=(3-d)F+\hbox{ quadratic terms in $t$}. \label{WDVV2}
\end{align}
In other words, the function $F$ solves the WDVV equation
\eqref{WDVV1}--\eqref{WDVV2}, and its third-order derivatives
$c_{\al\beta\gamma}$ are called the $3$-point correlator functions
in topological field theory.

Conversely, given a solution $F$ of \eqref{WDVV1}--\eqref{WDVV2}
(including a flat metric, a unity and a Euler vector field), one can
recover the structure of a Frobenius manifold.

A Frobenius manifold $M$ is said to be \emph{semisimple} if the
Frobenius algebras $T_t M$ are semisimple for generic points $t\in
M$.

On the Frobenius manifold $M$, its cotangent space $T_t^*M$ also
carries a Frobenius algebra structure, with an invariant bilinear
form $<dt^\al,dt^\beta>^*=\eta^{\al\beta}$ and a product given by
\begin{equation}\label{}
d t^\al\cdot d t^\beta=c^{\al\beta}_\gamma d t^\gamma, \quad
c_{\gamma}^{\al\beta}=\eta^{\al\ep}c_{\ep\gamma}^{\beta}. \nn
\end{equation}
Let
\begin{equation}\label{}
g^{\al\beta}=i_{E}(dt^{\al}\cdot dt^{\beta}),\nn
\end{equation}
then $(dt^{\al},dt^{\beta})^*=g^{\al\beta}$ defines a symmetric
bilinear form, called the \emph{intersection form}, on $T_t^*M$.

The above two bilinear forms on  $T^*M$ compose a pencil
$g^{\al\beta}+s\,\eta^{\al\beta}$ of flat metrics with parameter $s
$,  hence they induces a bi-hamiltonian structure $\{~ , ~ \}_{2}+s
\{~ , ~ \}_{1}$ of hydrodynamic type \cite{DN} on the loop space
$\left\{S^1\to M\right\}$. Furthermore, on the loop space one can
choose a family of functions $\ta_{\al,p}(t)$ with $\al=1,2,\dots,n$
and $p\ge0$ such that
\begin{align}\label{recursion}
\ta_{\al,0}=\eta_{\al\beta}t^{\beta},\quad \ta_{\al,1}=\frac{\p
F}{\p t^{\al}}, \quad \frac{\p^2\ta_{\al,p}}{\p t^\lm\p
t^{\mu}}=c^{\ep}_{\lm\mu}\frac{\p\ta_{\al,p-1}}{\p t^{\ve}}\, \hbox{
for } \, p>1.
\end{align}
The \emph{principal hierarchy} associated to $M$ is the following
system of Hamiltonian equations
\begin{equation}\label{prflow}
\frac{\p t^\gamma}{\p T^{\al,p}}=\left\{t^\gamma(x),
\int\ta_{\al,p}\,d
x\right\}_1:=\eta^{\gamma\beta}\p_x\frac{\p\ta_{\al,p}}{\p
t^{\beta}}, \quad \al,\gamma=1,2,\dots,n; ~p\ge0,\nn
\end{equation}
in which $x$ is the coordinate of the loop $S^1$. This hierarchy can
be written in a bi-hamiltonian recursion form that is consistent
with \eqref{recursion} if certain nonresonant condition is
fulfilled.

As typical examples, the orbit space of each (extended affine) Weyl
group is endowed a semisimple Frobenius manifold structure \cite{Du,Du3,
DZ2}. In Particular, the Frobenius manifolds for Weyl groups have
principal hierarchies that coincide with the dispersionless limit of
Drinfeld-Sokolov hierarchies \cite{DS, DZ, DLZ}. For Frobenius
manifolds corresponding to extended affine Weyl groups \cite{DZ2},
so far as we know, only in the case of type A their principal
hierarchies are clear, which are the dispersionless limit of the
extended bigraded Toda hierarchies \cite{Ca-bt, CDZ}.

The picture for infinite-dimensional Frobenius manifold and
integrable hierarchies is similar \cite{CDM, Ra, WX2}, though
examples are much less by now. For instance, Carlet, Dubrovin and
Mertens' Frobenius manifold in \cite{CDM} is associated with an
extension of the dispersionless Toda lattice hierarchy \cite{CM}. In
consideration of the relation between the Toda lattice and the
extended bigraded Toda hierarchies, it is natural to consider how
infinite-dimensional Frobenius manifolds are connected to the
finite-dimensional ones for extended affine Weyl groups of type A.
Such a connection will be studied in a revision of the construction
of \cite{CDM}, as to be seen below.

\vskip 0.5ex

Let us state the main results of the present paper, with similar
notations as in \cite{CDM, WX2}.

Firstly, one has two sets of holomorphic functions on closed disks
of the Riemann sphere $\mathbb{C}\cup\{\infty\}$ as follows:
 \beq\begin{array}{l}
     \mathH^-=\left\{f(z)=\dsum_{i\ge0}f_i\,z^{-i}\mid f \hbox{ holomorphic on
     } |z|\ge1 \right\},\\
 \mathH^+=\left\{\hat{f}(z)=\dsum_{i\ge0}\hat{f}_i\,z^{i}\mid \hat{f} \hbox{ holomorphic on
     } |z|\le1 \right\}.\end{array} \nn
    \eeq
Here a holomorphic function on a closed subset $D$ of
$\mathbb{C}\cup\{\infty\}$ means that the function can be extended
analyticly to a neighborhood of $D$. For instance, given $f(z)\in
\mathH^+$, there is some $\rho>0$ such that $f(z)$ can be extended
to a holomorphic function on $|z|<1+\rho$.

We fix two arbitrary positive integers $N$ and $M$, and consider the
following coset
   \beq
    \tilde{\cM}_{N,M}=(z^{N}+z^{N-1}\mathH^-)\times z^{-M}\mathH^+.
   \eeq
Any element of this coset is written as
$\bm{a}(z)=\big(a(z),\hat{a}(z)\big)$, where
   \beq \label{WZ2.2}
      a(z)=z^{N}+\sum_{i\leq N-1}v_{i} z^{i}, \quad
   \hat{a}(z)=\sum_{j\geq -M}\hat{v}_j z^{j}.
   \eeq
With coordinates given by the coefficients $v_i$ and $\hat{v}_j$ in
the above expansions, the coset $\tilde{\cM}_{N,M}$ can be
considered as an infinite-dimensional manifold.

For $\bm{a}(z)=(a(z),\hat{a}(z))\in\tilde{\cM}_{N,M}$, we introduce
two functions:
\begin{align}\label{WZ2.12}
  \zeta(z)=a(z)-\hat{a}(z),\quad  l(z)=a(z)_{>0}+\hat{a}(z)_{\leq 0},
\end{align}
where the subscripts ``$>0$'' and ``$\leq 0$''mean the projections
of a Laurent series to its positive part and nonpositive part
respectively. Note that $\zeta(z)$ is considered to be defined
holomorphically on a neighborhood of the unit circle $S^1$, and $l(z)$
can be extended analytically to the punched complex plane
$\mathbb{C}\setminus\{0\}$.
\begin{de} Let $\cM_{N,M}$ be a submanifold of $\tilde{\cM}_{N,M}$
consisting of points $(a(z),\hat{a}(z))$ such that the following
conditions are satisfied:
\begin{itemize}
\item[(M1)] the function $\hat{a}(z)$ has a pole of order $M$ at $0$, namely,
$\hat{v}_{-M}\neq 0$;
\item[(M2)] at $|z|=1$,
\beq \label{M2}
 a(z)\hat{a}'(z)-a'(z)\hat{a}(z)\neq 0,
  \quad  \zeta'(z)\neq 0,\quad  l'(z)\neq 0;
\eeq
\item[(M3)]
the function $\zeta(z)|_{S^1}$ has winding number $1$ around 0,
which defines a biholomorphic map from $S^1$ to a simple smooth
curve $\Gamma$ around $\zeta=0$.
\end{itemize}\end{de}

On $\cM_{N,M}$ we introduce a set of variables
\begin{equation}\label{flat}
\bm{t}\cup\bm{h}\cup\hat{\bm{h}}=
\{t^i\}_{i\in\mathZ}\cup\big\{h^j\big\}_{j=1}^{N}\cup\big\{\hat{h}^k\big\}_{k=0}^{M}
\end{equation}
by
\begin{align}\label{flatt}
&t^i=\frac{1}{2\pi\bm{i}}\oint_{|z|=1}\frac{\zeta(z)^{-i}}{i}\,\frac{dz}{z},
\quad i\in\mathbb{Z}\setminus \{0\};\quad
 t^0=\frac{1}{2\pi\bm{i}} \oint_{|z|=1}\log\frac{ z}{\zeta(z)}\,\frac{dz}{z};
 \\
& h^j=-\frac{N}{j} \res_{z=\infty} l(z)^{{j}/{N}}
\,\frac{dz}{z},\quad j=1,\cdots,N-1; \label{flath}
\\
& \hat{h}^0=\log \hat{v}_{-M};\quad
 \hat{h}^k=\frac{M}{k} \res_{z=0} l(z)^{{k}/{M}} \,\frac{dz}{z},\quad
 k=1,\cdots,M. \label{flathh}
\end{align}
\begin{Mthm}\label{main}
For any two positive integers $M$ and $N$, the infinite-dimensional
manifold $\cM_{N,M}$ is a semisimple Frobenius manifold with a
system of flat coordinates  \eqref{flat} such that
    \begin{itemize}
 \item[(i)] the unity vector field is
\begin{equation}\label{}
\bm{e}=\frac{\p}{\p \hat{h}^M};
\end{equation}
\item[(ii)] the potential $\mathcal{F}_{N,M}$ is
     \begin{align}\label{poten}
     \mathcal{F}_{N,M}=&\frac{1}{(2\pi
\bm{i})^2}\oint\oint_{|z_1|<|z_2|}\left(\frac{1}{2}\zeta(z_1)\zeta(z_2)-\zeta(z_1)l(z_2)
     +l(z_1)\zeta(z_2)\right)\times
     \nn\\
&\quad \times{\log\left(\frac{z_2-z_1}{z_2}\right)}
     \frac{dz_1}{z_1}\frac{dz_2}{z_2}
    -\frac{1}{(2\pi
\bm{i})^2}\oint_{|z|=1} \left(\frac{1}{2}\zeta(z)+l(z)\right)\frac{d
z}{z}\times
\nn\\
&\quad \times
\oint_{|z|=1}\zeta(z)\left(\log\frac{\zeta(z)}{z}-1\right)\frac{d
z}{z} +F_{N,M},
  \end{align}
  where $F_{N,M}$ is a function of $\bm{h}\cup\hat{\bm{h}}$ determined by
  \begin{align}\label{poten02}
    \frac{\p^3 F_{N,M}}{\p u\, \p v\, \p w}
    =-\big(\res_{z=\infty}+\res_{z=0}\big)\frac{\p_{u} l(z)\cdot\p_{v}
    l(z)\cdot\p_{w} l(z)}{z^2 l'(z)}dz
  \end{align}
  for any $u,v, w\in\bm{h}\cup\hat{\bm{h}}$;
\item[(iii)] the Euler vector field is
  \begin{align}\label{euler}
    \mathcal{E}_{N,M}= & -\sum_{i\in\mathZ} i\, t^i\frac{\p}{\p t^i}-\frac{N-1}{N}\frac{\p}{\p
t^0}+\sum_{j=1}^{N-1}\frac{j}{N}h^j\frac{\p}{\p h^j}+ \nn\\
&\quad
     +\sum_{k=1}^M\frac{k}{M}\hat{h}^k\frac{\p}{\p\hat{h}^k}
     +\frac{N+M}{N}\frac{\p}{\p\hat{h}^0},
\end{align}
and the charge of the Frobenius manifold is $d=1$.
    \end{itemize}
\end{Mthm}
This theorem will be proved in Section~2. There we will also write
down the flat metric and the intersection form for this semisimple
Frobenius manifold. Observe that the infinite-dimensional Frobenius
manifold $\cM_{1,1}$ is similar, but not the same, with the one
constructed in \cite{CDM} (the difference occurs from the definition
of flat metric, see \eqref{WZ2.14} below).

The function $F_{N,M}$ given by \eqref{poten02}, which depends
polynomially on $\bm{h}\cup\hat{\bm{h}}\cup\{ e^{\hat{h}^0}\}$, is
indeed the potential for the semisimple Frobenius manifold
$M(\tilde{A}_{M+N-1};N)$ on the orbit space of the extended Weyl
group $\tilde{W}^{(N)}(A_{M+N-1})$, see \cite{DZ2}. Since
$\cM_{N,M}$ is a space of the two functions $\zeta(z)$ and $l(z)$,
as $\zeta(z)\to0$ (the limit does not lie in $\cM_{N,M}$ but can be
considered as some kind of removable singularity), the Frobenius
structure on $\cM_{N,M}$ is reduced to the one on
$M(\tilde{A}_{M+N-1};N)$. This kind of limit agrees with the
reduction from the bi-hamiltonian structures for the dispersionless
Toda lattice hierarchy to that for the extended bigraded Toda
hierarchies, which are associated to the Frobenius manifolds
$\cM_{N,M}$ and $M(\tilde{A}_{M+N-1};N)$ respectively, see below and
\cite{Ca-bt, Wu}. It shall be indicated that, in this limit process
$M(\tilde{A}_{M+N-1};N)$ is not a Frobenius submanifold of
$\cM_{N,M}$ in the sense of Strachan \cite{St}.

In Section~3 we will show that the flat pencil on
$\mathT_{\bm{a}}^*\cM_{N,M}$ coincides with the one that induces the
dispersionless limit of the bi-hamiltonian structure obtained in
\cite{Wu} for the Toda lattice hierarchy. Accordingly the
dispersionless Toda lattice hierarchy is part of the principal
hierarchy associated to the Frobenius manifold $\cM_{N,M}$. The
complete principal hierarchy, in contrast to the two-component BKP
case \cite{WX2}, cannot be derived by applying only the
bi-hamiltonian recursion relation. Instead, as indicated by Carlet
and Mertens \cite{CM}, one needs to explicitly solve the flatness
equations for the deformed flat connection (the Levi-Civita
connection $\nabla$ deformed by the product of the Euler vector
field) on the infinite-dimensional Frobenius manifold. We shall
leave this problem open here.

The last section is devoted to conclusions and remarks.

\section{Construction of Frobenius manifolds}

We start to construct a semisimple Frobenius manifold structure on
the infinite-dimensional space
\[
\cM_{N,M}\subset(z^{N}+z^{N-1}\mathH^-)\times z^{-M}\mathH^+
\]
defined by the conditions (M1)--(M3) above with
arbitrary positive integers $M$ and $N$.

\subsection{Flat metric }

Recall that every element of $\cM_{N,M}$ has the form
\begin{equation}\label{aah}
\bm{a}=(a(z),\hat{a}(z))=\left(z^{N}+\sum_{i\leq N-1}v_{i} z^{i},
\sum_{j\geq -M}\hat{v}_j z^{j}\right).
\end{equation}

Let us describe the tangent and the cotangent bundles on
${\cM}_{N,M}$ with Laurent series. At a point $\bm{a}\in\cM_{N,M}$
we identify a vector $\p$ in the tangent space with its action $(\p
a(z), \p\hat{a}(z))$ on the ``point''. Hence the tangent space is
identified with a space of pairs of Laurent series as
 \beq \label{TaM}
  \mathT_{\bm{a}}\cM_{N,M}=z^{N-1}\mathH^-\times
  z^{-M}\mathH^+.
\eeq Clearly this space has a natural basis given by
\beq\label{WZ2.4} \frac{\p}{\p v_i}=(z^{i},0),\quad  i\leq N-1;
\qquad \frac{\p}{\p \hat{v}_j}=(0,z^{j}), \quad j\geq -M. \eeq
Accordingly, we write the cotangent space as \beq\label{WZ2.5}
  \mathT_{\bm{a}}^*\cM_{N,M}=z^{-N+1}\mathH^+\times
  z^{M}\mathH^-,
\eeq and the pairing of a covector
$\bm{\om}=(\om(z),\hat{\om}(z))\in\mathT_{\bm{a}}^*\cM_{N,M}$ with a
vector $\bm{X}=(X(z),\hat{X}(z))\in\mathT_{\bm{a}}\cM_{N,M}$ reads
\beq \label{WZ2.6}
  \la {\bm{\om},\bf X}\ra=\frac{1}{2\pi \bm{i}}\oint_{|z|=1}\big[\om(z)X(z)
+\hat{\om}(z)\hat{X}(z)\big]\dfrac{d z}{z}. \eeq Clearly, the
cotangent space has a dual basis with respect to \eqref{WZ2.4} as
follows
 \beq
  dv_i=(z^{-i},0),\quad  i\leq N-1; \quad d\hat{v}_j=(0,z^{-j}),\quad j\geq -M.
\eeq

We introduce two generating functions for covectors
\eqa
  && da(p):=\sum_{i\leq N-1}d v_i p^{i}=\left(\frac{p^N}{z^{N-1}(p-z)},0\right),
  \quad |z|<|p|,  \label{rpgener01} \\
  && d\hat{a}(p):=\sum_{j\geq -M}d\hat{v}_j p^{j}=\left(0,\frac{z^{M+1}}{p^{\,M}(z-p)}\right),
  \quad
  |z|>|p|.
  \label{rpgener02}
\eeqa
The Cauchy integral formula implies the following simple but useful lemma.
\begin{lem}
The following statements for the generating functions
\eqref{rpgener01}--\eqref{rpgener02} hold true:
\begin{itemize}
\item[(i)] for any vector $\bm{X}=(X(z),\hat{X}(z))\in\mathT_{\bm{a}}\cM_{N,M}$,
\begin{align}\label{daX}
  \la da(p),\bm{X}\ra=X(p),\quad \la
  d\hat{a}(p),\bm{X}\ra=\hat{X}(p);
\end{align}
\item[(ii)] for any covector
$\bm{\om}=(\om(z),\hat{\om}(z))\in\mathT^*_{\bm{a}}\cM_{N,M}$, we
have
\begin{align}
  &\bm{\om}=\frac{1}{2\pi\bm{i}}\oint_{|p|=1}
  \big[\om(p)da(p)+\hat{\om}(p)d\hat{a}(p)\big]\frac{d p}{p}. \label{gener}
\end{align}
\end{itemize}
\end{lem}

On $\mathT^*_{\bm{a}}\mathM_{N,M}$, we define a symmetric bilinear
form by \beq\label{WZ2.14}
  <d\al(p),d\beta(q)>^*=\frac{p\, q}{p-q} \left( \al'(p)-\beta'(q) \right),
\eeq where $\al'(p)={\p\al(p)}/{\p p}$ and
$\al,\beta\in\{a,\hat{a}\}$. The right hand side of \eqref{WZ2.14}
is understood as follows: if $\al$ and $\beta$ are the same, then
the denominator $(p-q)$ is eliminated by a factor of
$\al'(p)-\beta'(q)$; otherwise, suppose $\al=a$ and $\beta=\hat{a}$,
then $1/(p-q)$ is expanded to a series under the condition $|p|>|q|$
implied by \eqref{rpgener01}--\eqref{rpgener02}. This kind of
convention will be used throughout the present paper.

\begin{rem}
In the case $M=N=1$, this bilinear form used in \cite{CDM} (see
Equation~(1.16) therein) reads:
\begin{equation} \label{metricCDM}
 <d\al(p),d\beta(q)>^*=\frac{p\, q}{p-q} \left(\ep(\al)\al'(p)-\ep(\beta)\beta'(q) \right),
\end{equation}
where $\ep(a)=1$ and $\ep(\hat{a})=-1$. As the definition of
bilinear form is revised, the calculation below will be simplified.
\end{rem}

By using the nondegenerate pairing \eqref{WZ2.6}, one defines a
linear map \beq\label{WZ2.15} \eta:\
\mathT_{\bm{a}}^*\mathM_{N,M}\rightarrow \mathT_{\bm{a}}\mathM_{N,M}
\eeq such that
 \beq \label{WZ2.16} \la
\bm{\om}_1,\eta(\bm{\om}_2)\ra=<\bm{\om}_1,\bm{\om}_2>^*\eeq for any
covectors $\bm{\om}_1,\bm{\om}_2\in \mathT_{\bm{a}}^* \mathM_{N,M}$.

\begin{lem}\label{lem2.3}
The map $\eta$ defined in \eqref{WZ2.15}--\eqref{WZ2.16} can be
represented explicitly as
\begin{align}\label{}
 \eta(\bm{\om})(z)=&\Big( z a'(z)[\om(z)+\hat{\om}(z)]_{<0}-z[\om(z)a'(z)+\hat{\om}(z)\hat{a}'(z)]_{<0},\nn\\
 &\quad -z\hat{a}'(z)[\om(z)+\hat{\om}(z)]_{\geq
0}+z[\om(z)a'(z)+\hat{\om}(z)\hat{a}'(z)]_{\geq 0} \Big)
  \label{AWZ2.16}
\end{align}
with arbitrary $\bm{\om}=(\om(z),\hat{\om}(z)) \in \mathT_{\bm{a}}^*
\mathM_{N,M}$.  Moreover, the linear map $\eta$ is a bijection.
\end{lem}
\begin{proof} It follows from \eqref{daX} and \eqref{WZ2.16} that
\beq
  \eta(d\beta(q))=(<da(z),d\beta(q)>^*,<d\hat{a}(z),d\beta(q)>^*),\quad
  \beta\in\{a,\hat{a}\}.
\eeq Hence by using \eqref{gener} and $\la
\bm{\om}_1,\eta(\bm{\om}_2)\ra=\la \eta(\bm{\om}_1),\bm{\om}_2\ra$
we have \eqa
  \eta(\bm{\om})(z)&=&\frac{1}{2\pi\bm{i}}\oint_{|q|=1}
  \big[\om(q)\eta(da(q))+\hat{\om}(q)\eta(d\hat{a}(q))\big]\frac{d q}{q}\nn\\
  &=&\bigg(\frac{1}{2\pi\bm{i}}\oint_{|q|<|z|} <da(z), \om(q)da(q)+\hat{\om}(q)d\hat{a}(q)>^*\frac{d q}{q},\nn\\
  &~& \frac{1}{2\pi\bm{i}}\oint_{|q|>|z|} <d\hat{a}(z), \om(q)da(q)+\hat{\om}(q)d\hat{a}(q)>^*\frac{d q}{q}\bigg)\nn\\
  &=&\bigg(\frac{1}{2\pi\bm{i}}\oint_{|q|<|z|}\frac{q}{1-q/z}[a'(z)(\om(q)+\hat{\om}(q))-(\om(q)a'(q)+\hat{\om}(q)
  \hat{a}'(q))]\frac{d q}{q},\nn\\
  &~~& -\frac{1}{2\pi\bm{i}}\oint_{|q|>|z|} \frac{z}{1-z/q}
  [\hat{a}'(z)(\om(q)+\hat{\om}(q))-(\om(q)a'(q)+\hat{\om}(q)\hat{a}'(q))]\frac{d q}{q} \bigg)\nn\\
  &=&\Big(\, za'(z)[\om(z)+\hat{\om}(z)]_{<0}-z[\om(z)a'(z)+\hat{\om}(z)\hat{a}'(z)]_{<0},\nn\\
  &~~&-z\hat{a}'(z)[\om(z)+\hat{\om}(z)]_{\geq 0}+z[\om(z)a'(z)+\hat{\om}(z)\hat{a}'(z)]_{\geq 0} \, \Big).
\eeqa This is just \eqref{AWZ2.16}. Moreover, it implies the map
$\eta$ being surjective.

Next we want to show that the map $\eta$ is invertible. Suppose that
\eqa
\bm{X}=(X(z),\hat{X}(z))&=&\left(\dsum_{i\leq N-1}X_iz^i,
\dsum_{j\geq -M}\hat{X}_jz^j\right)
\in \mathT_{\bm{a}} \mathM_{N,M},\nn\\
 \bm{\om}=(\om(z),\hat{\om}(z))&=&\left(\dsum_{i\geq -N+1}\om_iz^i,
 \dsum_{j\leq M}\hat{\om}_jz^j\right)
\in \mathT_{\bm{a}}^* \mathM_{N,M}\nn \eeqa satisfy
$\bm{X}=\eta(\bm{\om})$, namely,
\begin{align}\label{} \label{Xom1}
&X(z)=z a'(z)[\om(z)+\hat{\om}(z)]_{<0}-z[\om(z)a'(z)+\hat{\om}(z)\hat{a}'(z)]_{<0},\\
&\hat{X}(z)=-z\hat{a}'(z)[\om(z)+\hat{\om}(z)]_{\geq
0}+z[\om(z)a'(z)+\hat{\om}(z)\hat{a}'(z)]_{\geq 0}. \label{Xom2}
\end{align}
It follows that \beq
X(z)-\hat{X}(z)=z(a'(z)-\hat{a}'(z))(\hat{\om}(z)_{<0}-{\om}(z)_{\geq
0}).\eeq Hence we obtain
\begin{equation}\label{WZ2.19}
\om(z)_{\geq
0}=-\left(\dfrac{X(z)-\hat{X}(z)}{z(a'(z)-\hat{a}'(z))}\right)_{\geq
0},\quad
\hat{\om}(z)_{<0}=\left(\dfrac{X(z)-\hat{X}(z)}{z(a'(z)-\hat{a}'(z))}\right)_{<0}.
\end{equation}
Here we have used the property $\zeta'(z)\ne0$ in the definition of
$\cM_{N,M}$.

On the one hand, by using \eqref{Xom1} and \eqref{Xom2} one has
\begin{align}\label{XXh}
& X(z)_{>0}=(z a'(z)[\om(z)+\hat{\om}(z)]_{<0})_{>0}
 =(z a'(z)_{>0}[\om(z)+\hat{\om}(z)]_{<0})_{>0},\\
& \hat{X}(z)_{\leq 0}=-(z\hat{a}'(z)[\om(z)+\hat{\om}(z)]_{\geq
0})_{\leq 0} =-(z \hat{a}'(z)_{<0}[\om(z)+\hat{\om}(z)]_{\geq
0})_{\leq 0}. \label{XXh2}
\end{align}
Writing
\[
\om(z)+\hat{\om}(z)=\dsum_{k\in \mathbb{Z}}\tilde{\om}_k z^k,
\]
then the equations \eqref{XXh}--\eqref{XXh2} can be rewritten as
\begin{equation}
\left(
  \begin{array}{c}
  X_{N-1} \\
  \vdots \\
  X_1\\
  \end{array}
\right)=K_{N-1}  \left( \begin{array}{c}
  \tilde{\om}_{-1} \\
  \vdots \\
  \tilde{\om}_{-N+1} \\
  \end{array}
\right),\quad \left(
  \begin{array}{c}
  \hat{X}_{-M} \\
  \vdots \\
  \hat{X}_{0} \\
  \end{array}
\right)=\hat{K}_{M+1}  \left( \begin{array}{c}
  \tilde{\om}_{0} \\
  \vdots \\
  \tilde{\om}_{M} \\
  \end{array}
\right),\label{WZ2.21}
\end{equation}
where \beq \label{KN} K_{N-1}=\left(
       \begin{array}{ccccc}
       N &  &  &  &  \\
        (N-1)v_{N-1}  & N &  &  &  \\
        (N-2)v_{N-2} & (N-1)v_{N-1} & N & & \\
       \vdots & \ddots & \ddots &  \ddots &  \\
       2\,v_2 & \cdots &  (N-2)v_{N-2} & (N-1)v_{N-1} & N \\
       \end{array}
       \right)
\eeq and \beq \label{KM} \hat{K}_{M+1}=\left(
       \begin{array}{ccccc}
       M\hat{v}_{-M} &  &  &  &  \\
        (M-1)\hat{v}_{-M+1}  & M\hat{v}_{-M} &  &  &  \\
       \vdots & \ddots & \ddots & ~ &  \\
       \hat{v}_{-1} &  &  \ddots & M\hat{v}_{-M} \\
       0 & \hat{v}_{-1} & \cdots & (M-1)\hat{v}_{-M+1} & M\hat{v}_{-M} \\
       \end{array}
       \right).
\eeq Both matrices $K_{N-1}$ and $\hat{K}_{M+1}$ are nondegenerate,
hence $\tilde{\om}_{-N+1},\cdots, \tilde{\om}_M$ can be solved from
\eqref{XXh}--\eqref{XXh2}. This fact together with \eqref{WZ2.19}
leads to that $\bm{\om}=(\om(z),\hat{\om}(z))$ is uniquely
determined by $\bm{X}=(X(z),\hat{X}(z))$. Therefore the lemma is
proved.
\end{proof}

With the help of the bijection $\eta$, the bilinear form
\eqref{WZ2.14} on the cotangent space induces a symmetric bilinear
form on $\mathT_{\bm{a}}\mathM_{N,M}$ as \beq\label{WZ2.24}
  <\p_1,\p_2>:=\la
  \eta^{-1}(\p_1),\p_2\ra=<\eta^{-1}(\p_1),\eta^{-1}(\p_2)>^*.
  \eeq
Recall the property \eqref{M2} for the functions
\begin{align}\label{zel}
  \zeta(z)=a(z)-\hat{a}(z),\quad  l(z)=a(z)_{>0}+\hat{a}(z)_{\leq
  0}.
\end{align}
One has the following lemma.
\begin{lem} The
bilinear form \eqref{WZ2.24} can be represented as follows: for any
tangent vectors $\p_1, \p_2\in\mathT_{\bm{a}}\mathM_{N,M}$,
\begin{align}\label{}
  <\p_1,\p_2>=&-\frac{1}{2\pi\bm{i}}\oint_{|z|=1}
  \frac{\p_1\zeta(z)\cdot\p_2\zeta(z)}{z^2\zeta'(z)}d z
  \nn\\
  &\quad -\res_{z=\infty}\frac{\p_1l(z)\cdot\p_2l(z)}{z^2l'(z)}d z
  -\res_{z=0}\frac{\p_1l(z)\cdot\p_2l(z)}{z^2 l'(z)}d z.
\label{WZ2.25}
\end{align}
  \end{lem}
 \begin{proof} Assume $\eta^{-1}(\p_1)=(\om(z),\hat{\om}(z))$, then we have
\eqa
  \la\eta^{-1}(\p_1),\p_2\ra
  &=&\frac{1}{2\pi\bm{i}}\oint_{|z|=1} [\om(z)\p_2 a(z)
  +\hat{\om}(z)\p_2\hat{a}(z)\big]\frac{d z}{z}\nn\\
  &=&\frac{1}{2\pi\bm{i}}\oint_{|z|=1} [(\om(z)_{\geq 0}-\hat{\om}(z)_{<0})
  (\p_2a(z)-\p_2\hat{a}(z))]\frac{d z}{z}\nn\\
  &&~+~\frac{1}{2\pi\bm{i}}\oint_{|z|=1}  [\om(z)+\hat{\om}(z)]_{<0}\p_2a(z)\frac{d z}{z}\nn\\
  &&~+~\frac{1}{2\pi\bm{i}}\oint_{|z|=1}  [\om(z)+\hat{\om}(z)]_{\geq 0}\p_2\hat{a}(z)\frac{d z}{z}.
  \label{WZ2.26}
\eeqa  The three integrals on the right hand side are denoted as
$I_1, I_2$ and $I_3$, respectively. Let us calculate them
separately.

Firstly, by using \eqref{zel} and $\p_1=(\p_1 a(z),\p_1
\hat{a}(z))$, the formulae \eqref{WZ2.19} read \beq
  \om(z)_{\geq 0}=-\left(\frac{\p_1\zeta(z)}{z \zeta'(z)}\right)_{\geq 0},\quad
  \hat{\om}(z)_{<0}=\left(\frac{\p_1\zeta(z)}{z\zeta'(z)}\right)_{<0}. \nn
\eeq Hence we have \beq I_1=-\frac{1}{2\pi\bm{i}}\oint_{|z|=1}
  \frac{\p_1\zeta(z)\cdot\p_2\zeta(z)}{z^2\zeta'(z)}dz.\nn\eeq

Secondly, recall
\[
l'(z)=N z^{N-1}+\sum_{i=1}^{N-1}i\,v_i z^{i-1}
-\sum_{j=1}^{M}j\,\hat{v}_{-j}z^{-j-1}.
\]
For simplicity, we denote $\Lambda(n)=(\delta_{i,j+1})_{n\times n}$ for positive
integer $n$ and $\Lambda(0)=0$. Observe that the matrix \eqref{KN}
is just \beq K_{N-1}=\left(\frac{l'(z)}{z^{N-1}}\right)_{z^{-1}\to
\Lambda(N-1)}.\eeq Let's take the following expansion near $z=\infty$
\beq \dfrac{1}{l'(z)} =\dfrac{1}{z^{N-1}}(f_0+f_{-1} z^{-1}+f_{-2}
z^{-2}+\cdots), \nn \eeq and set $f_i=0$ for $i>0$,  then we have
\beq K_{N-1}^{-1}=\left(\frac{z^{N-1}}{l'(z)}\right)_{z^{-1}\to
\Lambda(N-1)}=(f_{-i+j})_{(N-1)\times(N-1)}. \eeq Thus by using the
first relation in \eqref{WZ2.21}, we know \eqa
 I_2&=&(\p_1 v_1, \dots, \p_1 v_{N-1})K_{N-1}^{-1}\left(
                    \begin{array}{c}
                      \p_2 v_{N-1} \\
                      \vdots \\
                      \p_2 v_1 \\
                    \end{array}
                    \right)\nn\\
 &=&\sum_{i,j=1}^{N-1} \p_1 v_i\cdot f_{-i+j}\cdot \p_2 v_{N-j}\nn\\
 &=& -\res_{z=\infty}\left( \sum_{i,j=1}^{N-1}\p_1v_i \, z^{i}\,\cdot f_{-i+j}z^{-i+j} \cdot \p_2v_{N-j} z^{N-j}
  \cdot z^{-N-1}\right) dz \nn\\
 &=& -\res_{z=\infty}\frac{\p_1 l(z)\cdot\p_2 l(z)}{z^2\, l'(z)} dz.\nn
\eeqa

Thirdly, in the same way as above one has  \beq
\hat{K}_{M+1}=-\left(l'(z)z^{M+1}\right)_{z \to \Lambda(M+1)} \eeq
 and  \beq
\hat{K}_{M+1}^{-1}=-\left(\frac{1}{l'(z)z^{M+1}}\right)_{z\to
\Lambda(M+1)}=-(g_{i-j})_{(M+1)\times(M+1)}, \eeq where $g_m$ are
given by the following expansion near $z=0$: \beq \dfrac{1}{l'(z)}
={z^{M+1}}(g_0+g_{1} z+g_{2} z^{2}+\cdots) , \nn \eeq and $g_m=0$
for $m<0$. The second equation in \eqref{WZ2.21} leads to \eqa
 I_3&=&
(\p_1 \hat{v}_0, \dots, \p_1 \hat{v}_{-M})\hat{K}_{M+1}^{-1}\left(
                    \begin{array}{c}
                      \p_2 \hat{v}_{-M} \\
                      \vdots \\
                      \p_2 \hat{v}_0 \\
                    \end{array}
                    \right)\nn\\
 &=&-\sum_{i,j=0}^{M} \p_1 \hat{v}_{-i}\cdot g_{i-j}\cdot \p_2 \hat{v}_{j-M}\nn\\
 &=& -\res_{z=0}\left( \sum_{i,j=0}^{M} \p_1 \hat{v}_{-i} \, z^{-i}\,\cdot g_{i-j}z^{i-j} \cdot
 \p_2\hat{v}_{j-M} z^{j-M} \cdot z^{M-1}\right) dz \nn\\
 &=& -\res_{z=0}\frac{\p_1 l(z)\cdot\p_2 l(z)}{z^2\, l'(z)} dz.\nn
\eeqa Thus we complete the proof of this lemma.
\end{proof}

Next we want to choose a system of ``flat'' coordinates for the
bilinear form \eqref{WZ2.25}. According to the definition of
$\cM_{N,M}$, one can consider the inverse function of $\zeta(z)$
such that \beq z=z(\zeta):\Gamma \longrightarrow S^1, \eeq and this
function is holomorphic on a neighborhood of the curve $\Gamma$
surrounding $\zeta=0$. We assume the following Riemann-Hilbert
factorization on the $\zeta$-plane \beq z(\zeta) = f_0^{-1}(\zeta)
f_\infty(\zeta)\quad \mbox{for} \quad \zeta \in \Gamma, \eeq where
the functions $f_0(\zeta)$ and ${f_\infty(\zeta)}/{\zeta}$ are
holomorphic inside and outside the curve $\Gamma$, respectively. Let
us fix this factorization by normalizing
$$
f_\infty(\zeta) = \zeta +O(1), \quad |\zeta|\to\infty.
$$
Now we have Taylor expansions with coefficients $t^i$:
\begin{align}\label{}
& \log f_0(\zeta) =-\dsum_{i\geq 0} t^i \zeta^i , \quad  |\zeta| \to
0; \\
& \log\frac{f_\infty(\zeta)}{\zeta}=\dsum_{i\geq
1}\frac{t^{-i}}{\zeta^i},\quad |\zeta|\to\infty. \label{WZ2.34}
\end{align}
In other words, \beq
t^i=\frac{1}{2\pi\bm{i}}\oint_{\Gamma}\zeta^{-i-1}\log\frac{z(\zeta)}{\zeta}
d\zeta. \label{WZ2.35}\eeq

On the other hand, recall $l(z)=z^{N}+\sum_{i=1}^{N-1}v_{i}
z^{i}+\sum_{j=0}^{M}\hat{v}_{-j} z^{-j}$ and denote \eqa
\chi(z):=l(z)^{{1}/{N}}\ \mbox{ near }\ \infty, \quad
\hat{\chi}(z):=l(z)^{{1}/{M}}\ \mbox{ near }\ 0. \eeqa Let $z(\chi)$
and $z(\hat{\chi}) $ be the inverse functions of $\chi(z)$ and
$\hat{\chi}(z)$ respectively. Their logarithms can be expanded as
follows:
\begin{align}\label{}
&\log z(\chi)=\log\chi-\frac{1}{N}({h^1}\chi^{-1}+\cdots
  +{h^{N-1}}\chi^{-N+1})+O(\chi^{-N}), \quad
  |\chi|\to\infty; \label{WZ2.38} \\
&\log
z(\hat{\chi})=-\log\hat{\chi}+\frac{1}{M}(\hat{h}^0+{\hat{h}^1}\hat{\chi}^{-1}+\cdots
  +{\hat{h}^{M}}\hat{\chi}^{-M})+O(\chi^{-M-1}),\quad |\hat\chi|\to\infty, \label{WZ2.39}
\end{align}
where in particular $\hat{h}^{0}=\log \hat{v}_{-M}$.

Observe that the variables \beq
 \bm{t}\cup\bm{h}\cup\hat{\bm{h}}=\{t^i\mid i\in\mathZ\}\cup\{h^j\mid j=1,\ldots,N-1\}
 \cup\{\hat{h}^k\mid k=0,\ldots,M\},\label{WZ2.40}
\eeq determine $\zeta(z)$ and $l(z)=\chi(z)^{N}=\hat{\chi}(z)^{M}$
uniquely, hence also \beq a(z)=l(z)+\zeta(z)_{\leq 0},\quad
\hat{a}(z)=l(z)-\zeta(z)_{>0}.\label{WZ2.41}\eeq It means that
\eqref{WZ2.40} is indeed a system of coordinates on the manifold
$\mathM_{N,M}$.

\begin{prop} \label{thm-inpd}
The above coordinates $\bm{t}\cup\bm{h}\cup\hat{\bm{h}}$ can be
defined equivalently by \eqref{flatt}--\eqref{flathh}.  The bilinear
form $<~,~>$ in \eqref{WZ2.25} satisfies
 \begin{align}
   &<\frac{\p}{\p t^{i_1}},\frac{\p}{\p t^{i_2}}>=-\delta_{i_1+i_2,-1},
   \quad  i_1,i_2\in\mathZ \label{inpdflat01};\\
   &<\frac{\p}{\p h^{j_1}},\frac{\p}{\p h^{j_2}}>=\frac{1}{N}\delta_{j_1+j_2,N},
   \quad j_1,j_2\in\{1,2,\ldots,N-1\}; \label{inpdflat02}\\
   &<\frac{\p}{\p \hat{h}^{k_1}},\frac{\p}{\p \hat{h}^{k_2}}>=\frac{1}{M}\delta_{k_1+k_2,M},
   \quad k_1,k_2\in\{0,1,\dots,M\} \label{inpdflat03}
 \end{align}
and any other pairing between these vector fields vanishes.
Consequently, the bilinear form $<~,~>$ is a nondegnerate flat
metric on $\mathT_{\bm{a}}\mathM_{N,M}$ with flat coordinates
$\bm{t}\cup\bm{h}\cup\hat{\bm{h}}$.
\end{prop}

The proof of this proposition is mainly based on following lemma.
\begin{lem}\label{thm-dzhh}
The coordinates $\bm{t}\cup\bm{h}\cup\hat{\bm{h}}$ satisfy
\begin{align}\label{}
& \frac{\p \zeta(z)}{\p t^i}=-z \zeta^i(z)\zeta'(z),\label{WZ2.43} \\
& \frac{\p l(z)}{\p h^j}=  (z \chi(z)^{N-j-1} \chi'(z))_{>0}, \label{lh} \\
& \frac{\p l(z)}{\p\hat{h}^k}= -
(z\hat{\chi}(z)^{M-k-1}\hat{\chi}'(z))_{\le0}, \label{lhh}
\end{align}
and that the following derivatives vanish:  \beq  \dfrac{\p
\zeta(z)}{\p h^j}= \dfrac{\p \zeta(z)}{\p \hat{h}^k}= \dfrac{\p
l(z)}{\p t^i}=\left(\dfrac{\p l(z)}{\p h^j}\right)_{\leq 0}=
\left(\dfrac{\p l(z)}{\p \hat{h}^k}\right)_{>0}=0.\label{WZ2.42}\eeq
\end{lem}
\begin{proof}
By using
\begin{align}\label{}
\zeta^i=\frac{\p}{\p
t^i}\log\frac{z(\zeta)}{\zeta}=\frac{1}{z(\zeta)}\frac{\p
z(\zeta)}{\p t^i}
\end{align}
and the ``thermodynamical identity"
\begin{align}\label{}
\left.\frac{\p \zeta(z)}{\p
t^i}\right|_{z=z(\zeta)}+\left.\zeta'(z)\right|_{z=z(\zeta)}
\frac{\p z(\zeta)}{\p t^i}=0,
\end{align}
we obtain \eqref{WZ2.43}. Similarly, by using
\begin{align*}
& \dfrac{\p z(\chi)}{\p
h^j}=-z(\chi)\left(\dfrac{1}{N}\chi^{-j}+O(\chi^{-N})\right), \quad |\chi|\to\infty; \\
& \dfrac{\p z(\hat{\chi})}{\p
\hat{h}^k}=z(\hat{\chi})\left(\dfrac{1}{M}\hat{\chi}^{-k}
+O(\hat{\chi}^{-M-1})\right), \quad |\hat\chi|\to\infty
\end{align*}
and the corresponding ``thermodynamical identities", we derive
\begin{align}\label{}
&\frac{\p \chi(z)}{\p h^j}=z\chi'(z)
\left(\dfrac{1}{N}\chi(z)^{-j}+O(z^{-N})\right), \quad |z|\to\infty;
 \label{WZ2.45} \\
 & \dfrac{\p \hat{\chi}(z)}{\p \hat{h}^k}=-z\hat{\chi}'(z) \left(\dfrac{1}{M}\hat{\chi}(z)^{-k}+O(z^{M+1})\right),
  \quad |z|\to 0.
\end{align}
Thus the equalities \eqref{lh}--\eqref{lhh} are obtained. The
derivatives \eqref{WZ2.42} being zero is trivial. Therefore the
lemma is proved.
\end{proof}

\begin{prfof}{Proposition~\ref{thm-inpd}}
The first assertion follows from a short calculation. For the second
assertion, with $i_1,~i_2\in\mathbb{Z}$ we have
 \eqa <\frac{\p}{\p t^{i_1}},\frac{\p}{\p
t^{i_2}}>&=&-\frac{1}{2\pi\bm{i}}\oint_{|z|=1}
  \dfrac{\frac{\p \zeta(z)}{\p t^{i_1}}\cdot\frac{\p \zeta(z)}{\p t^{i_2}}}{z^2\zeta'(z)}dz\nn\\
  &=& -\frac{1}{2\pi\bm{i}}\oint_{|z|=1} \zeta^{i_1+i_2}(z)\zeta'(z) dz\nn\\
  &=& -\frac{1}{2\pi\bm{i}}\oint_{\Gamma}\zeta^{i_1+i_2}d\zeta=-\delta_{i_1+i_2,-1}. \eeqa
For $ 1\leq j_1,~j_2\leq N-1$, \eqa && <\frac{\p}{\p
h^{j_1}},\frac{\p}{\p h^{j_2}}>=-\res_{z=\infty}
  \dfrac{\frac{\p l(z)}{\p h^{j_1}}\cdot\frac{\p l(z)}{\p h^{j_2}}}{z^2l'(z)}dz\nn\\
  &=&-\res_{z=\infty}\frac{\chi(z)^{N-1}}{N}\Big(\chi(z)^{-j_1}+O(z^{-N})\Big)
  \Big(\chi(z)^{-j_2}+O(z^{-N})\Big)\chi'(z) dz\nn\\
  &=& -\res_{\chi=\infty}\frac{\chi^{N-1-j_1-j_2}}{N} d\chi=\frac{1}{N}\delta_{j_1+j_2,N}. \eeqa
Finally, for $0\le k_1,~k_2\leq M$,
 \eqa && <\frac{\p}{\p \hat{h}^{k_1}},\frac{\p}{\p
\hat{h}^{k_2}}>=-\res_{z=0}
  \dfrac{\frac{\p l(z)}{\p \hat{h}^{k_1}}\cdot\frac{\p l(z)}{\p \hat{h}^{k_2}}}{z^2l'(z)}dz\nn\\
  &=&-\res_{z=0}\frac{\hat{\chi}(z)^{M-1}}{M}\Big(\hat{\chi}(z)^{-k_1}+O(z^{M+1})\Big)
  \Big(\hat{\chi}(z)^{-k_2}+O(z^{M+1})\Big)\hat{\chi}'(z) dz\nn\\
  &=& -\res_{\hat{\chi}=\infty}\frac{\hat{\chi}^{M-1-k_1-k_2}}{M} d\hat{\chi}=\frac{1}{M}\delta_{k_1+k_2,M}. \nn\eeqa
Clearly all other pairings between these vectors vanish. The
proposition is proved.
\end{prfof}


\subsection{Frobenius algebra structure 
} In order to endow a Frobenius algebra structure on the tangent
space of $\cM_{N,M}$, let us introduce a product on the cotangent
space $\mathT^*_{\bm{a}}\mathM_{N,M}$ first:
 \beq
d\al(p)\cdot d\beta(q) = \frac{p\, q}{p-q} \Big( \al'(p)\, d\beta(q)
-\beta'(q)\, d\al(p)\Big) \label{WZ2.49} \eeq with
$\al'(p)={\p\al(p)}/{\p p}$ and $\al,\beta\in\{a,\hat{a}\}$. Note
that the dominator on the right hand side has the same meaning as in
\eqref{WZ2.14}.

\begin{lem}\label{lem2.6} On $\mathT^*_{\bm{a}}\mathM_{N,M}$ the following assertions
hold true:
\begin{itemize}
\item[(i)] the multiplication defined by \eqref{WZ2.49} is associative and
commutative; more generally, for $\al_i\in\{a,\hat{a}\}$ one has
  \begin{align}
  d\al_1(p_1)\cdot d\al_2(p_2)\cdot\cdots\cdot d\al_k(p_k)=\sum_{i=1}^{k}\left(\prod_{j\neq
  i}\frac{\al_j'(p_j)}{p_i^{-1}-p_j^{-1}}\right)d\al_i(p_i).\label{WZ2.50}
  \end{align}
\item[(ii)] the bilinear form $<\ ,\ >^*$ given in \eqref{WZ2.14}
  is invariant with respect to the above multiplication.
\end{itemize}
\end{lem}
\begin{proof} The commutativity of \eqref{WZ2.49} is obvious. The formula
\eqref{WZ2.50} can be verified by an induction, which yields the
associativity of the multiplication. Hence we deduce the first
assertion.

The second assertion follows from
\begin{align}\label{}
 <d\al(p)\cdot
d\beta(q),d\gamma(r)>^*
 =&\frac{\al'(p)\beta'(q)}{(p^{-1}-r^{-1})(r^{-1}-q^{-1})}
  +\frac{\beta'(q)\gamma'(r)}{(q^{-1}-p^{-1})(p^{-1}-r^{-1})}\nn\\
  &\quad +\frac{\gamma'(r)\al'(p)}{(r^{-1}-q^{-1})(q^{-1}-p^{-1})}
  \nn\\
  =&<d\beta(q)\cdot d\gamma(r),d\al(p)>^*\nn
\end{align}
with $\al, \beta, \gamma\in\{a,\hat{a}\}$.
 \end{proof}

\begin{lem}\label{lem2.7} The multiplication \eqref{WZ2.49} can be
represented in Laurent series as
\begin{align}
  \bm{\om}_1\cdot\bm{\om}_2=&
  z\Big(\big[\om_2(z)\big(\om_1(z)a'(z)\big)_{\geq 0} -\om_2(z)\big(\hat{\om}_1(z)\hat{a}'(z)\big)_{<0}\nn\\
  &-\om_1(z)\big(\om_2(z)a'(z)\big)_{<0} -\om_1(z)\big(\hat{\om}_2(z)\hat{a}'(z)\big)_{<0}\big]_{\geq -N},\nn\\
  &\big[\hat{\om}_2(z)\big(\om_1(z)a'(z)\big)_{\geq 0} +\hat{\om}_2(z)\big(\hat{\om}_1(z)\hat{a}'(z)\big)_{\geq 0}\nn\\
  &+\hat{\om}_1(z)\big(\om_2(z)a'(z)\big)_{\geq 0}
  -\hat{\om}_1(z)\big(\hat{\om}_2(z)\hat{a}'(z)\big)_{<0}\big]_{\leq M-1}\Big)\label{WZ2.51}
\end{align}
for any cotangent vectors
$\bm{\om}_i=(\om_i(z),\hat{\om}_i(z))\in\mathT_{\bm{a}}^*\mathM_{N,M}$
with $i=1, 2$. Moreover, this multiplication has unity \beq
\bm{e}^*:=\left(0, \frac{z^M}{M\hat{v}_{-M}}\right)=\frac{1}{M}d
\hat{h}^{0}. \label{WZ2.52} \eeq
\end{lem}
\begin{proof}
The equality \eqref{WZ2.51} is verified by using \eqref{gener}.
Secondly, for any cotangent vector
$\bm{\om}=(\om(z),\hat{\om}(z))\in\mathT_{\bm{a}}^*\mathM_{N,M}$, in
consideration of the form of $(a(z), \hat{a}(z))$ in  \eqref{aah} we
have
 \eqa \bm{e}^*\cdot\bm{\om} &=&
z\left(\left[-\om(z)\left(\frac{z^M}{M\hat{v}_{-M}}
\hat{a}'(z)\right)_{<0}\right]_{\geq -N}\, ,\,
\left[-\hat{\om}(z)\left(\frac{z^M}{M\hat{v}_{-M}} \hat{a}'(z)\right)_{<0}\right]_{\leq M-1}\right)\nn\\
&=& z\left(\left(\frac{\om(z)}{z}\right)_{\geq -N},
\left(\frac{\hat{\om}(z)}{z}\right)_{\leq M-1}\right)=\bm{\om}.\nn
 \eeqa
Thus we complete the proof of the lemma.
\end{proof}

As a combination of Lemmas~\ref{lem2.6} and \ref{lem2.7}, we obtain
\begin{prop} The cotangent space  $\mathT_{\bm{a}}^*\mathM_{N,M}$ carries
a Frobenius algebra structure with the multiplication defined in
\eqref{WZ2.49}, the unity $\bm{e}^*$ given in \eqref{WZ2.52} and the
non-degenerate invariant bilinear form \eqref{WZ2.14}.
\end{prop}

Due to the bijection $\eta$ in \eqref{WZ2.5}, we achieve the
following result.
\begin{cor}\label{cor-FM}
The tangent space $\mathT_{\bm{a}}\mathM_{m,n}$ is a Frobenius
algebra with multiplication between any vectors $\bm{X}_1$ and
$\bm{X}_2$ defined by \beq
  \bm{X}_1\cdot\bm{X}_2:=\eta\big(\eta^{-1}(\bm{X}_1)\cdot\eta^{-1}(\bm{X}_2)\big),
\eeq  the unity vector
\begin{align}
  \bm{e}:=\eta\big(\bm{e}^*\big)=\big(1,1\big)=\dfrac{\p}{\p \hat{h}^M}
\end{align}
and the invariant inner product \eqref{WZ2.25}.
\end{cor}


\subsection{The potential}

Now we proceed to compute the following symmetric $3$-tensor
\begin{equation}\label{3tensor}
c(\p_u,\p_v, \p_w)= <\p_u\cdot\p_v, \p_w>, \quad u, v,
w\in\bm{t}\cup\bm{h}\cup\hat{\bm{h}}
\end{equation}
where for simplicity we write $\p_u={\p}/{\p u}\in
\mathT_{\bm{a}}\mathM_{N,M}$. In the present section, unless
otherwise stated the following convention will be assumed:
\[
 i, i_1, i_2, i_3\in\mathZ,\quad j, j_1,j_2,j_3\in\{1,2,\ldots,N-1\},
\quad k,k_1,k_2,k_3\in\{0,1,\ldots,M\}.
\]
\begin{lem}
The symmetric 3-tensor \eqref{3tensor} is given by
\begin{align}\label{}
& <\p_{t^{i_1}}\cdot \p_{h^{j_2}},\p_{\hat{h}^{k_3}}>=0, \\
 & <\p_{t^{i_1}}\cdot \p_{t^{i_2}},\p_{h^{j_3}}>
  =\frac{-1}{2\pi \bm{i}}\oint_{|z|=1}z \Big(\zeta'(z)\zeta(z)^{i_1+i_2}\Big)_{<-1}
  \Big({\chi}'(z)\chi(z)^{N-1-j_3}\Big)_{\geq 0}   ~ dz,
   \\
 & <\p_{h^{j_1}}\cdot \p_{h^{j_2}},\p_{t^{i_3}}>
  =\frac{-1}{2\pi N \bm{i}}\oint_{|z|=1}z
  \Big(\zeta'(z)\zeta(z)^{i_3}\Big)_{<0}\Big({\chi}'(z)\chi(z)^{N-1-j_1-j_2}\Big)_{\geq -1} ~ dz, \\
 & <\p_{h^{j_1}}\cdot \p_{h^{j_2}},\p_{\hat{h}^{k_3}}> \nn\\
  &\qquad =\frac{-1}{2\pi N \bm{i}}\oint_{|z|=1}z \Big({\chi}'(z)\chi(z)^{N-1-j_1-j_2}\Big)_{\geq -1}
  \Big(\chi'(z)\chi(z)^{M-1-k_3}\Big)_{<0} ~ dz, \\
 & <\p_{t^{i_1}}\cdot \p_{t^{i_2}},\p_{\hat{h}^{k_3}}>
  =\frac{1}{2\pi  \bm{i}}\oint_{|z|=1}z
  \Big(\zeta'(z)\zeta(z)^{i_1+i_2}\Big)_{\geq -1}\Big(\hat{\chi}'(z)\hat{\chi}(z)^{M-1-k_3}\Big)_{<0} ~dz,
   \\
 & <\p_{\hat{h}^{k_1}}\cdot \p_{\hat{h}^{k_2}},\p_{t^{i_3}}>
  =\frac{-1}{2\pi M \bm{i}}\oint_{|z|=1}z
  \Big(\zeta'(z)\zeta(z)^{i_3}\Big)_{\geq 0}\Big(\hat{\chi}'(z)\hat{\chi}(z)^{M-1-k_1-k_2}\Big)_{<-1} ~ dz,
   \\
  & <\p_{\hat{h}^{k_1}}\cdot \p_{\hat{h}^{k_2}},\p_{{h}^{j_3}}>
  \nn\\
 &\qquad  =\frac{-1}{2\pi M \bm{i}}\oint_{|z|=1}z \Big({\chi}'(z)\chi(z)^{N-1-j_3}\Big)_{\geq 0}
  \Big(\hat{\chi}'(z)\hat{\chi}(z)^{M-1-k_1-k_2}\Big)_{<-1} ~ dz,
  \\
& <\p_{h^{j_1}}\cdot \p_{h^{j_2}},\p_{h^{j_3}}>= \frac{-1}{2\pi
N^2\bm{i}}\oint_{|z|=1}z\zeta'(z)_{<0}
  \Big(\chi'(z)\chi(z)^{N-1-j_1-j_2-j_3}\Big)_{\geq -1}dz
   \nn\\
  &\qquad -\res_{z=\infty}\frac{z\Big(\chi'(z)\chi(z)^{N-1-j_1}\Big)_{\geq 0}
  \Big(\chi'(z)\chi(z)^{N-1-j_2}\Big)_{\geq 0} \Big(\chi'(z)\chi(z)^{N-1-j_3}\Big)_{\geq 0}}
  {l'(z)}dz,
  \\
 & <\p_{\hat{h}^{k_1}}\cdot\p_{\hat{h}^{k_2}},\p_{\hat{h}^{k_3}}>= \frac{1}{2\pi
M^2\bm{i}}\oint_{|z|=1}z\zeta'(z)_{\geq 0}
  \Big(\hat{\chi}'(z)\hat{\chi}(z)^{M-1-k_1-k_2-k_3}\Big)_{<-1}dz
  \nn\\
  &\quad +\res_{z=0}\frac{z\Big(\hat{\chi}'(z)\hat{\chi}(z)^{M-1-k_1}\Big)_{<0}
  \Big(\hat{\chi}'(z)\hat{\chi}(z)^{M-1-k_2}\Big)_{<0} \Big(\hat{\chi}'(z)\hat{\chi}(z)^{M-1-k_3}\Big)_{<0}}
  {l'(z)}dz, \label{hhat3}
  \\
 &  <\p_{t^{i_1}}\cdot\p_{t^{i_2}},\p_{t^{i_3}}>
  =\frac{1}{2\pi\bm{i}}\oint_{|z|=1}z\zeta'(z)\zeta(z)^{i_1+i_2+i_3}\Big(\frac{1}{2}\Pi[\zeta'(z)]
  -l'(z)\Big)dz\nn\\
  &\qquad-\frac{1}{4\pi\bm{i}}\oint_{|z|=1}{z}\zeta'(z)\Big(\zeta(z)^{i_1+i_2}\Pi[\zeta'(z)\zeta(z)^{i_3}]
  \nn\\
  &\qquad + \zeta(z)^{i_2+i_3}\Pi[\zeta'(z)\zeta(z)^{i_1}]+\zeta(z)^{i_3+i_1}\Pi[\zeta'(z)\zeta(z)^{i_2}]
  \Big)dz,
\end{align}
with the mapping $\Pi[f(z)]=f(z)_{\geq 0}-f(z)_{<0}$ in the final
equality.
\end{lem}

\begin{proof}
We denote
\[
\eta_{ u v}=<\p_u,\p_v>,
\]
of which the values are given in Proposition~\ref{thm-inpd}, and then we have
  \eqa \label{3tens}
  <{\p_u}\cdot{\p_v},{\p_w}>=\eta_{u\tilde{u}}\eta_{v\tilde{v}}\eta_{w\tilde{w}}
  <d\tilde{u}\cdot d\tilde{v},d\tilde{w}>^*
  =\eta_{u\tilde{u}}\eta_{v\tilde{v}}\eta_{w\tilde{w}}\la d\tilde{u}
    \cdot d\tilde{v},\eta(d\tilde{w})\ra
  \eeqa
with $u, v, w,\tilde{u},\tilde{v},\tilde{w}
\in\bm{t}\cup\bm{h}\cup\hat{\bm{h}}$. Let us compute the data in the
right hand side.

First of all, one can check the following equalities
\begin{align}\label{WZ2.55}
dt^i=&\Big(-(\zeta(z)^{-i-1})_{\geq -N+1},~(\zeta(z)^{-i-1})_{\leq
M} \Big),
 \\
dh^j=&\Big((\chi(z)^{j-N})_{\geq -N+1},~0\Big), \label{dh}
   \\
d\hat{h}^k=&\Big(0, ~(\hat{\chi}(z)^{k-M})_{\leq M}\Big).
\label{dhh}
\end{align}
Indeed, it follows from \eqref{flatt} that
\[
\frac{\p t^i}{\p v_r}=-\frac{1}{2\pi\bm{i}}\oint_{|p|=1}
\zeta(p)^{-i-1} p^{r-1} d p,\quad \frac{\p t^i}{\p
\hat{v}_s}=\frac{1}{2\pi\bm{i}}\oint_{|p|=1} \zeta(p)^{-i-1} p^{s-1}
d p,
\]
hence
 \eqa dt^i&=&\dsum_{r\leq N-1}\frac{\p t^i}{\p v_r} dv_r
+\dsum_{s\geq -M}\frac{\p t^i}{\p \hat{v}_s} d\hat{v}_s\nn\\
 &=&\Big(\dsum_{r\leq N-1}\frac{\p t^i}{\p v_r} z^{-r},
   \dsum_{s\geq -M}\frac{\p t^i}{\p \hat{v}_s}z^{-s}\Big)\nn\\
&=& \Big(-(\zeta(z)^{-i-1})_{\geq -N+1},~(\zeta(z)^{-i-1})_{\leq M}
\Big).\nn\eeqa Similarly, the equalities \eqref{dh}--\eqref{dhh} can
be checked with the help of
\begin{align*}\label{}
&\frac{\p h^j}{\p v_r}=-\res_{p=\infty}\chi(p)^{j-N}p^{r-1}d p,\quad
\frac{\p h^j}{\p \hat{v}_s}=0, \nn\\
& \frac{\p \hat{h}^k}{\p
\hat{v}_s}=\res_{p=0}\hat{\chi}(p)^{k-M}p^{s-1}d p, \quad \frac{\p
\hat{h}^k}{\p {v}_r}=0
\end{align*}
implied by \eqref{flath}--\eqref{flathh}.

Secondly, by using \eqref{WZ2.51} we have
\begin{align}\label{}
& dt^{i_1}\cdot dt^{i_2} =z\Big( \Big[ a'(z)\zeta(z)^{-i_1-i_2-2}
-\zeta(z)^{-i_1-1}(\zeta(z)^{-i_2-1}\zeta'(z))_{<0}\nn\\
&\qquad
-\zeta(z)^{-i_2-1}(\zeta(z)^{-i_1-1}\zeta'(z))_{<0}\Big]_{\geq -N}\,
,\, -\Big[\hat{a}'(z)\zeta(z)^{-i_1-i_2-2} \nn\\
& \qquad +\zeta(z)^{-i_1-1}(\zeta(z)^{-i_2-1}\zeta'(z))_{\geq 0}
 +\zeta(z)^{-i_2-1}(\zeta(z)^{-i_1-1}\zeta'(z))_{\geq 0}\big]_{\leq M-1}
\Big)
\end{align}
with
\[
a'(z)=\zeta'(z)_{<0}+l'(z), \quad \hat{a}'(z)=-\zeta'(z)_{\geq
0}+l'(z),
\]
and
\begin{align}
 & dh^{j_1}\cdot dh^{j_2}=\Big(N z\Big[\chi'(z)\chi(z)^{j_1+j_2-N-1}
-\chi(z)^{j_1-N}\Big(\chi'(z)\chi(z)^{j_2-1}\Big)_{<0}\nn\\
& \qquad
-\chi(z)^{j_2-N}\Big(\chi'(z)\chi(z)^{j_1-1}\Big)_{<0}\Big]_{\geq-N}\,
,\, 0\Big), \label{WZ2.57} \\
& d\hat{h}^{k_1}\cdot d\hat{h}^{k_2}=\Big(0\,,\, M z
\Big[-\hat{\chi}'(z)
\hat{\chi}(z)^{k_1+k_2-M-1}+\hat{\chi}(z)^{k_1-M}\Big(\hat{\chi}'(z)\hat{\chi}(z)^{k_2-1}\Big)_{\geq
0}\nn\\
&\qquad
+\hat{\chi}(z)^{k_2-M}\Big(\hat{\chi}'(z)\hat{\chi}(z)^{k_1-1}\Big)_{\geq
0}\Big]_{\le M-1}\Big).
\end{align}

Thirdly, substituting \eqref{WZ2.55}--\eqref{dhh} into
\eqref{AWZ2.16} we obtain
\begin{align}\label{etadt}
&\eta(dt^i)=z\Big(\big(\zeta'(z)\zeta(z)^{-i-1}\big)_{<0},
-\big(\zeta'(z)\zeta(z)^{-i-1}\big)_{\geq 0}\Big),  \\
&\eta(dh^j)=N z\Big(\big(\chi'(z)\chi(z)^{j-1}\big)_{\geq 0},
\big(\chi'(z)\chi(z)^{j-1}\big)_{\geq 0}\Big),  \\
 & \eta(d\hat{h}^k)=-M z\Big( \big(\hat{\chi}'(z)\hat{\chi}(z)^{k-1}\big)_{<0},
\big(\hat{\chi}'(z)\hat{\chi}(z)^{k-1}\big)_{<0}\Big).
\label{dtadhh}
\end{align}

Finally, we substitute the above data into \eqref{3tens}, and
conclude the lemma after a tedious but straightforward computation.
\end{proof}

\begin{lem}[\cite{DZ2}] There exists a function $F_{N,M}$ depending polynomially
on $\bm{h}\cup\hat{\bm{h}}\cup\{e^{\hat{h}^0}\}$ such that
\begin{align}\label{poten03}
    \frac{\p^3 F_{N,M}}{\p u\, \p v\, \p w}
    =-\big(\res_{z=\infty}+\res_{z=0}\big)\frac{\p_{u} l(z)\cdot\p_{v}
    l(z)\cdot\p_{w} l(z)}{z^2\, l'(z)}dz
  \end{align}
for any $u,v, w\in\bm{h}\cup\hat{\bm{h}}$.
\end{lem}
This lemma can also be proved in the same way as in the appendix of
\cite{WX2}, with the replacement $\p_z\mapsto z\p_z$ and $d z\mapsto
d z/z$.

Let us introduce a function of $\bm{t}\cup\bm{h}\cup\hat{\bm{h}}$ as
\begin{align}\label{}
&\mathcal{F}_{N,M}=\frac{1}{(2\pi
\bm{i})^2}\oint\oint_{|z_1|<|z_2|}\left(\frac{1}{2}\zeta(z_1)\zeta(z_2)-\zeta(z_1)l(z_2)
     +l(z_1)\zeta(z_2)\right)\times\nn\\
     &\quad\times \, \log\left(\frac{z_2-z_1}{z_2}\right)
     \frac{dz_1}{z_1}\frac{dz_2}{z_2} +\left(\frac{1}{2}t^{-1}-\hat{h}^M\right)\sum_{i\ge0}t^i t^{-i-1}
     +F_{N,M}, \label{cFMN}
\end{align}
where
\[
\log\left(\frac{z_2-z_1}{z_2}\right)=-\dsum_{m\ge1}\dfrac{1}{m}\left(\frac{z_1}{
z_2}\right)^m.
\]

\begin{prop} \label{thm-Fc}
The function $\mathcal{F}_{N,M}$ satisfies
\begin{equation}
c(\p_u,\p_v,\p_w)=\frac{\p^3 \mathcal{F}_{N,M}}{\p u\,\p v\,\p
w},\quad u,v,w\in\bm{t}\cup\bm{h}\cup\hat{\bm{h}}.
\end{equation}
Consequently, letting $\nabla$ denote the Levi-Civita connection of
the metric $<~,~>$, the $4$-tensor
  $\nabla_{\p_s}c(\p_u,\p_v,\p_w)$ is symmetric with respect to the flat
  coordinates $s$, $u$, $v$ and $w$.
  \end{prop}
\begin{proof}
In this proof we use the notation $f^*(z)=z f'(z)$ for a
differentiable function $f(z)$. If $f(z)$ and $g(z)$ are holomorphic
(single-valued unless otherwise stated) on a neighborhood of
$|z|=1$, then one has
\[
f^*(z)_{\ge0}=f^*(z)_{>0}=z f'(z)_{>0}, \quad
f^*(z)_{\le0}=f^*(z)_{<0}=z f'(z)_{<0}
\]
and the following formula of integration by parts:
\begin{align}\label{}
\frac{1}{2\pi \bm{i}}\oint_{|z|=1} f^*(z)g(z)\frac{d z}{z}=
-\frac{1}{2\pi \bm{i}}\oint_{|z|=1} f(z)g^*(z)\frac{d z}{z}.
\end{align}

Let
\begin{equation}\label{}
Y_i(z)=\left\{\begin{array}{cl}
         \dfrac{\zeta(z)^{i+1}}{i+1}, & i\in\mathbb{Z}\setminus\{-1\}; \\ \\
         \log\zeta(z), & i=-1.
       \end{array}\right.
\end{equation}
Clearly all $Y_i(z)$ and $Y_i^*(z)$ are holomorphic at $|z|=1$
except $Y_{-1}(z)$ being multi-valued, hence one has
\[
\frac{1}{2\pi \bm{i}}\oint_{|z|=1}Y_i^*(z)\frac{d
z}{z}=\frac{\delta_{i,-1}}{2\pi
\bm{i}}\oint_{|z|=1}\frac{\zeta'(z)}{\zeta(z)}d z=\delta_{i,-1}.
\]
By using Lemma~\ref{thm-dzhh}, it is easy to see
\begin{align}
&\frac{\p\zeta(z)}{\p t^i}=-z\zeta'(z)\zeta(z)^i=-Y_i^*(z), \\
&\frac{\p^2\zeta(z)}{\p t^{i_1}\p t^{i_2}}=Y_{i_1+i_2}^{**}(z),
\quad \frac{\p^3\zeta(z)}{\p t^{i_1}\p t^{i_2}\p
t^{i_3}}=-Y_{i_1+i_2+i_3}^{***}(z).
\end{align}

We want to compute the third order derivatives of $\mathcal{F}$ with
respect to $\bm{t}\cup\bm{h}\cup\hat{\bm{h}}$. To this end, let
$\tilde{\mathcal{F}}$ denote the integral part on the right hand
side of \eqref{cFMN}, then we have
\begin{align}
&\frac{\p^3\tilde{\mathcal{F}}}{\p t^{i_1}\,\p t^{i_2}\,\p
t^{i_3}}\nn\\
 =&\frac{1}{(2\pi \bm{i})^2}\oint\oint_{|z_1|<|z_2|}
 \bigg(-\frac{1}{2}Y_{i_1+i_2+i_3}^{***}(z_1)\zeta(z_2)-
 \frac{1}{2}\zeta(z_1)Y_{i_1+i_2+i_3}^{***}(z_2) \nn\\
 &\qquad -\frac{1}{2}\sum_{\mathrm{c.p.}(i_1,i_2,i_3)}\left(
 Y_{i_1+i_2}^{**}(z_1)Y_{i_3}^{*}(z_2)+Y_{i_1}^{*}(z_1)Y_{i_2+i_3}^{**}(z_2)
 \right) \nn\\
 &\qquad +Y_{i_1+i_2+i_3}^{***}(z_1)l(z_2)-l(z_1)Y_{i_1+i_2+i_3}^{***}(z_2)\bigg)
 \log\left(\frac{z_2-z_1}{z_2}\right)\frac{dz_1}{z_1}\frac{dz_2}{z_2} \nn\\
=&\frac{1}{(2\pi \bm{i})^2}\oint\oint_{|z_1|<|z_2|}
 \bigg(-\frac{1}{2}Y_{i_1+i_2+i_3}^{**}(z_1)\zeta(z_2)+
 \frac{1}{2}\zeta(z_1)Y_{i_1+i_2+i_3}^{**}(z_2) \nn\\
 &\qquad -\frac{1}{2}\sum_{\mathrm{c.p.}(i_1,i_2,i_3)}\left(
 Y_{i_1+i_2}^{*}(z_1)Y_{i_3}^{*}(z_2)-Y_{i_1}^{*}(z_1)Y_{i_2+i_3}^{*}(z_2)
 \right) \nn\\
 &\qquad +Y_{i_1+i_2+i_3}^{**}(z_1)l(z_2)+l(z_1)Y_{i_1+i_2+i_3}^{**}(z_2)\bigg)\frac{z_1}{z_2-z_1}
 \frac{dz_1}{z_1}\frac{dz_2}{z_2} \nn\\
=&\frac{1}{2\pi \bm{i}}\oint
 \bigg(-\frac{1}{2}Y_{i_1+i_2+i_3}^{**}(z_2)_{<0}\zeta(z_2)+
 \frac{1}{2}\zeta(z_2)_{<0}Y_{i_1+i_2+i_3}^{**}(z_2) \nn\\
 &\qquad -\frac{1}{2}\sum_{\mathrm{c.p.}(i_1,i_2,i_3)}\left(
 Y_{i_1+i_2}^{*}(z_2)_{<0}Y_{i_3}^{*}(z_2)-Y_{i_1}^{*}(z_2)_{<0}Y_{i_2+i_3}^{*}(z_2)
 \right) \nn\\
 &\qquad +Y_{i_1+i_2+i_3}^{**}(z_2)_{<0}l(z_2)+l(z_2)_{<0}Y_{i_1+i_2+i_3}^{**}(z_2)\bigg)
 \frac{dz_2}{z_2} \nn\\
=&\frac{1}{2\pi \bm{i}}\oint
 \bigg(Y_{i_1+i_2+i_3}^{*}(z)\left(\frac{1}{2}\zeta^*(z)_{>0}-
 \frac{1}{2}\zeta^*(z)_{<0}-l^*(z)_{>0}-l^*(z)_{<0} \right) \nn\\
 &\qquad -\frac{1}{2}\sum_{\mathrm{c.p.}(i_1,i_2,i_3)}
 Y_{i_1+i_2}^{*}(z)(Y_{i_3}^{*}(z)_{>0}-Y_{i_3}^{*}(z)_{<0}) \bigg)
 \frac{d z}{z} \nn\\
=&\frac{1}{2\pi \bm{i}}\oint
 \bigg(Y_{i_1+i_2+i_3}^{*}(z)\left(\frac{1}{2}\Pi\zeta^*(z)-l^*(z) \right) \nn\\
 &\qquad -\frac{1}{2}\sum_{\mathrm{c.p.}(i_1,i_2,i_3)}
 Y_{i_1+i_2}^{*}(z) z\Pi\,Y_{i_3}'(z) \bigg)
 \frac{d z}{z}  \nn\\
 &-\frac{1}{2}\sum_{\mathrm{c.p.}(i_1,i_2,i_3)}\frac{1}{2\pi \bm{i}}\oint
 Y_{i_1+i_2}^{*}(z)\frac{d z}{z}\cdot \frac{1}{2\pi \bm{i}}\oint Y_{i_3}^{*}(z)\frac{d
 z}{z}\nn\\
=&<\p_{t^{i_1}}\cdot \p_{t^{i_2}},\p_{t^{i_3}}>
-\frac{1}{2}\sum_{\mathrm{c.p.}(i_1,i_2,i_3)}\delta_{i_1+i_2,-1}\delta_{i_3,-1},
\end{align}
where ``c.p.'' means ``cyclic permutation'', in the second and the
fourth equalities the formula of integration by parts is used. Hence
in consideration of $\p_{t^i}F_{N,M}=0$ we obtain
\begin{align}\label{}
\frac{\p^3\mathcal{F}_{N,M} }{\p t^{i_1}\,\p t^{i_2}\,\p
t^{i_3}}=&\frac{\p^3\tilde{\mathcal{F}} }{\p t^{i_1}\,\p t^{i_2}\,\p
t^{i_3}}+\frac{1}{2}\sum_{\mathrm{c.p.}(i_1,i_2,i_3)}\delta_{i_1+i_2,-1}\delta_{i_3,-1}
\nn\\
=&c(\p_{t^{i_1}}, \p_{t^{i_2}}, \p_{t^{i_3}}).
\end{align}
In the same way,
\begin{align}
&\frac{\p^3\tilde{\mathcal{F}}}{\p t^{i_1}\,\p t^{i_2}\,\p
\hat{h}^{k_3}}\nn\\
 =&\frac{1}{(2\pi \bm{i})^2}\oint\oint_{|z_1|<|z_2|}
 \bigg(Y_{i_1+i_2}^{**}(z_1)(z_2\hat{\chi}'(z_2)\hat{\chi}(z_2)^{M-k_3-1})_{\le0}
 \nn\\
 &\qquad
 -(z_1\hat{\chi}'(z_1)\hat{\chi}(z_1)^{M-k_3-1})_{\le0}Y_{i_1+i_2}^{**}(z_2)\bigg)
 \log\left(\frac{z_2-z_1}{z_2}\right)\frac{dz_1}{z_1}\frac{dz_2}{z_2}
 \nn\\
 =&\frac{1}{(2\pi \bm{i})^2}\oint\oint_{|z_1|<|z_2|}
 \bigg( Y_{i_1+i_2}^{*}(z_1)(z_2\hat{\chi}'(z_2)\hat{\chi}(z_2)^{M-k_3-1})_{\le0}
 \nn\\
 &\qquad
 +(z_1\hat{\chi}'(z_1)\hat{\chi}(z_1)^{M-k_3-1})_{\le0}Y_{i_1+i_2}^{*}(z_2)\bigg)\frac{z_1}{z_2-z_1}
 \frac{dz_1}{z_1}\frac{dz_2}{z_2} \nn\\
 =&\frac{1}{2\pi \bm{i}}\oint
 \bigg( Y_{i_1+i_2}^{*}(z_2)_{<0}(z_2\hat{\chi}'(z_2)\hat{\chi}(z_2)^{M-k_3-1})_{\le0}
  \nn\\
 &\qquad
 +(z_2\hat{\chi}'(z_2)\hat{\chi}(z_2)^{M-k_3-1})_{<0}Y_{i_1+i_2}^{*}(z_2)\bigg)\frac{dz_2}{z_2}
 \nn\\
 =&\frac{1}{2\pi \bm{i}}\oint
 Y_{i_1+i_2}^{*}(z)_{>0}(z\hat{\chi}'(z)\hat{\chi}(z)^{M-k_3-1})_{<0}\frac{d
 z}{z}
\nn\\
=&\frac{1}{2\pi \bm{i}}\oint
 Y_{i_1+i_2}^{*}(z)_{\ge0}(z\hat{\chi}'(z)\hat{\chi}(z)^{M-k_3-1})_{\le0}\frac{d
 z}{z}
 \nn\\
 &\qquad
 -\frac{1}{2\pi \bm{i}}\oint
 Y_{i_1+i_2}^{*}(z)\frac{d
 z}{z}\cdot  \frac{1}{2\pi \bm{i}}\oint z\hat{\chi}'(z)\hat{\chi}(z)^{M-k_3-1}\frac{d
 z}{z}
\nn\\
=&<\p_{t^{i_1}}\cdot \p_{t^{i_2}},\p_{\hat{h}^{k_3}}>
+\delta_{i_1+i_2,-1}\delta_{k_3,M},
\end{align}
which leads to
\[
\frac{\p^3\mathcal{F}_{N,M} }{\p t^{i_1}\,\p t^{i_2}\,\p
\hat{h}^{k_3}}=c(\p_{t^{i_1}}, \p_{t^{i_2}}, \p_{\hat{h}^{k_3}}).
\]
The other cases are similar. Therefore the proposition is proved.
 \end{proof}


\begin{prop}\label{thm-FF}
The function $\mathcal{F}_{N,M}$ can be written in a coordinate free
way as
\begin{align}
\mathcal{F}_{N,M}=&\frac{1}{(2\pi
\bm{i})^2}\oint\oint_{|z_1|<|z_2|}\left(\frac{1}{2}\zeta(z_1)\zeta(z_2)-\zeta(z_1)l(z_2)
     +l(z_1)\zeta(z_2)\right)\times\nn\\
     &\times\, {\log\left(\frac{z_2-z_1}{z_2}\right)}
     \frac{dz_1}{z_1}\frac{dz_2}{z_2}
    -\frac{1}{(2\pi
\bm{i})^2}\oint_{|z|=1} \left(\frac{1}{2}\zeta(z)+l(z)\right)\frac{d
z}{z}\times \nn\\
&\times\,
\oint_{|z|=1}\zeta(z)\left(\log\frac{\zeta(z)}{z}-1\right)\frac{d
z}{z} +F_{N,M}. \label{FMN}
\end{align}
\end{prop}

\begin{proof}
According to \eqref{flatt}--\eqref{flath} one has
\[
\frac{1}{2\pi \bm{i}}\oint
\left(\frac{1}{2}\zeta(z)+l(z)\right)\frac{d
z}{z}=-\frac{1}{2}t^{-1}+\hat{h}^M.
\]
Comparing \eqref{FMN} with \eqref{cFMN}, we only need to show
\begin{equation}\label{}
\frac{1}{2\pi
\bm{i}}\oint\zeta(z)\left(\log\frac{\zeta(z)}{z}-1\right)\frac{d
z}{z}=\sum_{i\ge0}t^i t^{-i-1}.
\end{equation}
Indeed, starting from the right hand side,
\begin{align}\label{}
\mathrm{r.h.s.}=&\frac{1}{2\pi
\bm{i}}\oint_\Gamma\frac{1}{2}\left(\log\frac{z(\zeta)}{\zeta}\right)^2 d\zeta \nn\\
 = & -\frac{1}{2\pi
\bm{i}}\oint_{|z|=1}\zeta(z)\log\frac{z}{\zeta(z)}
\left(\frac{1}{z}-\frac{\zeta'(z)}{\zeta(z)}\right)d z \nn\\
=&\frac{1}{2\pi
\bm{i}}\oint_{|z|=1}\zeta(z)\log\frac{\zeta(z)}{z}\frac{d z}{z}+
\frac{1}{2\pi \bm{i}}\oint_\Gamma \log\frac{z(\zeta)}{\zeta}d
\zeta \nn\\
=&\frac{1}{2\pi
\bm{i}}\oint_{|z|=1}\zeta(z)\log\frac{\zeta(z)}{z}\frac{d z}{z}+
t^{-1}=\mathrm{l.h.s.}
\end{align}
Thus the proposition is proved.
\end{proof}


\subsection{The Euler vector field}

We are to fix a Euler vector field on $\cM_{N,M}$ and show the
quasi-homogeneity property of $\mathcal{F}_{N,M}$, which will imply
that $\cM_{N,M}$ is really a Frobenius manifold.

Let us assign a degree to each of the flat coordinates:
 \begin{align}\label{degree}
   &\deg\,t^i=-i,\quad i\in\mathZ\setminus\{0\}; \qquad \deg\,e^{t^0}=\frac{1}{N}-1; \\
   & \deg\,h^j=\frac{j}{N},\quad
   1\le j\le N-1; \\
   &  \deg\,e^{\hat{h}^0}=1+\frac{M}{N}; \quad
   \deg\,\hat{h}^k=\frac{k}{M}, \quad 1\le k\le M \label{degree3}
 \end{align}
and introduce the following vector field
   \begin{align}
    \mathcal{E}_{N,M}= & -\sum_{i\in\mathZ\setminus\{0\}} i\, t^i\frac{\p}{\p t^i}-\frac{N-1}{N}\frac{\p}{\p
t^0}+\sum_{j=1}^{N-1}\frac{j}{N}h^j\frac{\p}{\p h^j}+ \nn\\
&\quad
     +\sum_{k=1}^M\frac{k}{M}\hat{h}^k\frac{\p}{\p\hat{h}^k}
     +\Big(1+\frac{M}{N}\Big)\frac{\p}{\p\hat{h}^0}. \label{euler2}
\end{align}

\begin{prop} \label{EF}
The function $\mathcal{F}_{N,M}$ satisfies
\begin{equation}\label{eurerF}
Lie_{
\mathcal{E}_{N,M}}\mathcal{F}_{M,M}=2\mathcal{F}_{N,M}+\hbox{quadratic
terms in flat coordinates}.
\end{equation}
\end{prop}
\begin{proof}
If we assume $\deg\,z=1/N$ besides \eqref{degree}--\eqref{degree3},
then each of the functions $\zeta(z)$ and $l(z)$ are homogeneous of
degree $1$, hence $\mathcal{F}_{N,M}$ given in \eqref{FMN} has
degree $2$ modulo quadratic terms. The proposition is proved.
\end{proof}

Up to now,  we have shown that $\cM_{N,M}$ is a Frobenius manifold,
on which the potential is $\mathcal{F}_{N,M}$, the unity vector
filed is $\bm{e}=\p/\p\hat{h}^M$ and the Euler vector field is
$\mathcal{E}_{N,M}$.

From the proof of the above proposition, one also sees \beq
 \left(\mathcal{E}_{N,M}+ \frac{z}{N}\frac{\p}{\p z}\right) \vartheta(z)=\vartheta(z) \label{WZ2.75}\eeq
for $\vartheta(z) \in \{a(z), \hat{a}(z), \zeta(z), l(z)\}$. Hence
 $\mathcal{E}_{N,M}$ is equal to the following vector filed
\beq \mathcal{E}_{N,M}=\dsum_{r\leq N-1} \left(1-\frac{r}{N}\right)
v_r\frac{\p}{\p v_r} +\dsum_{s\geq -M} \left(1-\frac{s}{N}\right)
\hat{v}_s\frac{\p}{\p \hat{v}_s}. \label{WZ2.74}\eeq By using
 \eqref{WZ2.4}, one can represent  $\mathcal{E}_{N,M}$ to the form
 of Laurent series as
\beq \mathcal{E}_{N,M}= \left(a(z)-\frac{z}{N} a'(z),
\hat{a}(z)-\frac{z}{N}\hat{a}'(z)\right). \label{WZ2.72} \eeq We
remark that one can start from the formula \eqref{WZ2.72}, then
deduce \eqref{euler2} and \eqref{eurerF} with the help of relevant
result in \cite{DZ2}.

\subsection{Proof of Main Theorem}
According to Corollary~\ref{cor-FM}, Propositions~\ref{thm-inpd} and
\ref{thm-Fc}--\ref{EF}, one sees that $\mathcal{M}_{N,M}$ is a
Frobenius manifold with properties required by the Main Theorem in
Section~1. To prove the Main Theorem, we only need to show the
semisimplicity of $\mathcal{M}_{N,M}$.

With the same method as in \cite{CDM}, let
\begin{equation}
du(p)=\frac{\hat{a}'(p)}{\zeta'(p)}{d
a(p)}-\frac{a'(p)}{\zeta'(p)}{d \hat{a}(p)}, \quad p\in S^1
\end{equation}
which is a generating function for a basis of the cotangent space
$\mathT^*_{\bm{a}}\mathM_{N,M}$. By using \eqref{WZ2.14} and
\eqref{WZ2.49}, one can check
\begin{align}\label{}
&<du(p),du(q)>^*=f(p)\delta(p-q), \\
&du(p)\cdot du(q)=f(p)\delta(p-q)\,du(p),  \label{WZ2.79}
\end{align}
 where
\[
 f(p)=-p^2\dfrac{a'(p)\hat{a}'(p)}{\zeta'(p)}
\]  and
$\delta(p-q)=\sum_{k\in\mathZ}\left(p^{k}/q^{k+1}\right)$ is the
formal delta function such that
\[
\frac{1}{2\pi \bm{i}}\oint_{|q|=1} f(q)\delta(p-q) d q =f(p).
\]
The formula \eqref{WZ2.79} implies that at a generic points there is
no nilpotent elements in the Frobenius algebra
$T^*_{\bm{a}}\cM_{N,M}$, hence the Frobenius algebra
$T^*_{\bm{a}}\cM_{N,M}$ is semisimple and so is
$T_{\bm{a}}\cM_{N,M}$. That is to say, the Frobenius manifold
$\cM_{N,M}$ is semisimple. We thus complete the proof of
the Main Theorem. \qquad \qquad $\Box$

\subsection{The intersection form}

Now let us consider the intersection form of $\cM_{N,M}$, which is
closely related to the theory of integrable hierarchies. Being
analogous to finite-dimensional cases, we define the intersection
form on the cotangent space $\mathT^*_{\bm{a}}\mathM_{m,n}$ by
  \begin{align}\label{WZ2.80}
  \big(d\al(p),d\beta(q)\big)^*:=i_{\mathcal{E}_{N,M}}\big(d\al(p)\cdot d\beta(q)\big),
  \quad \al,\beta\in\big\{a,\hat{a}\big\}.
  \end{align}

\begin{prop}
For $ \al,\beta\in\big\{a,\hat{a}\big\}$ it holds that \beq
\big(d\al(p), d\beta(q)\big)^*=\frac{p\, q}{p-q} \Big(
\al'(p)B(\beta(q))-B(\al(p))\beta'(q)\Big),\label{WZ2.81}\eeq
where $B(\alpha(p))=\alpha(p)-\dfrac{p}{N}\alpha'(p)$.
\end{prop}
\begin{proof} Substituting \eqref{WZ2.49} into \eqref{WZ2.80} and using \eqref{WZ2.75},  we
have
 \eqa
  \big(d\alpha(p),d\beta(q)\big)^*
  &=&\la\mathcal{E}_{N,M}, d\alpha(p) \cdot d \beta(q)\ra\nn\\
  &=& \frac{p\, q}{p-q}
  \left(\al'(p) \la \mathcal{E}_{N,M}, d\beta(q)\ra -\beta'(q)
  \la\mathcal{E}_{N,M}, d\al(p)\ra \right)\nn\\
  &=& \frac{p\, q}{p-q} \left(\al'(p) \left[\beta(q)-\frac{q}{N} \beta'(q)\right]
  -\beta'(q)\left[\alpha(p)-\frac{p}{N} \alpha'(p)\right]\right),\nn\\
  &=&\frac{p\, q}{p-q}
\left[ \al'(p)\beta(q) -\al(p)\beta'(q)\right] +\frac{p\, q}{N}
\al'(p) \beta'(q)\nn \eeqa which is exactly \eqref{WZ2.81}. The
proposition is proved.
\end{proof}

The intersection form, in the same way as for the flat metric
\eqref{WZ2.14}, induces a linear map \beq\label{WZ2.82} g:\
\mathT_{\bm{a}}^*\mathM_{N,M}\rightarrow \mathT_{\bm{a}}\mathM_{N,M}
\eeq  such that \beq \label{WZ2.83} \la
\bm{\om}_1,g(\bm{\om}_2)\ra=(\bm{\om}_1,\bm{\om}_2)^*\eeq for any
$\bm{\om}_1,\bm{\om}_2\in \mathT_{\bm{a}}^* \mathM_{N,M}$.

\begin{lem}
The linear map $g$ defined by \eqref{WZ2.82}--\eqref{WZ2.83} is a
bijection.
\end{lem}

\begin{proof} The proof is very similar with that of Lemma~\ref{lem2.3}.

On the one hand, according to \eqref{daX} one has \beq
  g(d\beta(q))=\Big((d\beta(q),da(z))^*,(d\beta(q),d\hat{a}(z))^*\Big),\quad
  \beta\in\{a,\hat{a}\}.\nn
\eeq Hence for arbitrary $\bm{\om}=(\om(z),\hat{\om}(z)) \in
\mathT_{\bm{a}}^* \mathM_{m,n}$, by using \eqref{gener} we have
\begin{align}\label{}
g(\bm{\om})(z)=&\frac{1}{2\pi\bm{i}}\oint_{|q|=1}
  \big[\om(q)g(da(q))+\hat{\om}(q)g(d\hat{a}(q))\big]\frac{d q}{q}\nn\\
=&\Big(z a'(z)\big[\omega(z) B(a(z))+\hat{\omega}(z)B(\hat{a}(z))\big]_{<0} \nn\\
&\quad  -
zB(a(z))\big[\omega(z)a'(z)+\hat{\omega}(z)\hat{a}'(z)\big]_{<0}\, ,
\nn\\
&\quad -z \hat{a}'(z)\big[\omega(z)
B(a(z))+\hat{\omega}(z){B}(\hat{a}(z))\big]_{\geq 0} \nn\\
&\quad +z B(\hat{a}(z))[\omega(z)a'(z)
+\hat{\omega}(z)\hat{a}'(z)]_{\geq 0} \Big)
\end{align}
which implies that $g$ is surjective.

On the other hand, given any $\bm{X}=(X(z),\hat{X}(z))\in
\mathT_{\bm{a}} \mathM_{N,M}$, we want to solve the equation
  \beq \bm{X}=g(\bm{\om}),\nn \eeq
i.e., \eqa
 X(z)&=&z a'(z)\big[\omega(z) B(a(z))+\hat{\omega}(z)B(\hat{a}(z))\big]_{<0}\nn\\
 &~&\qquad- zB(a(z)) \big[\omega(z)a'(z)+\hat{\omega}(z)\hat{a}'(z)\big]_{<0},\nn\\
\hat{X}(z)&=&-z \hat{a}'(z)\big[\omega(z) B(a(z))
+\hat{\omega}(z){B}(\hat{a}(z))\big]_{\geq 0}\nn\\
&~&\qquad+z B(\hat{a}(z))[\omega(z)a'(z)
+\hat{\omega}(z)\hat{a}'(z)]_{\geq 0}.\nn
 \nn\eeqa
Denote
\begin{equation}\label{}
\Theta(z)=\dfrac{\hat{a}'(z)X(z)-a'(z)\hat{X}(z)}{z[a(z)\hat{a}'(z)
-a'(z)\hat{a}(z)]}.
\end{equation}
It is straightforward to check \eqa \Theta(z)
=(\omega(z)a'(z))_{\geq 0}-(\hat{\omega}(z)\hat{a}'(z))_{<0}.
\label{Taom} \eeqa
Observe \beq \frac{1}{a'(z)}= \left(\frac{1}{a'(z)}\right)_{\leq
-N+1},\quad
\frac{1}{\hat{a}'(z)}=\left(\frac{1}{\hat{a}'(z)}\right)_{\geq
M+1}\nn\eeq then by using \eqref{Taom} we obtain
\begin{align}\label{omTa}
&\omega(z)=\omega(z)_{>-N}
=\left(\frac{1}{a'(z)}(\omega(z)a'(z))_{\geq 0}\right)_{>-N}
=\left(\frac{1}{a'(z)}\Theta(z)_{\geq 0}\right)_{>-N},\\
&\hat{\omega}(z)=\hat{\omega}(z)_{\leq M}
=\left(\frac{1}{\hat{a}'(z)}(\hat{\omega}(z)\hat{a}'(z))_{<0}\right)_{\leq
M} =-\left(\frac{1}{\hat{a}'(z)}\Theta(z)_{<0}\right)_{\leq M}.
\label{omhTa}
\end{align}
It follows that $g$ is injective. The lemma is proved.
\end{proof}

With the help of the bijection $g$, a bilinear form on the tangent
space $T_{\bm{a}}\mathM_{N,M}$ can be defined as
\begin{align}
  \big(\p_1,\p_2\big):=\la g^{-1}(\p_1),\p_2\ra=\big(g^{-1}(\p_1),g^{-1}(\p_2)\big)^*.
\end{align}
By using \eqref{omTa}--\eqref{omhTa}, a short computation leads to

\begin{prop}  For any $\p_1, \p_2\in T_{\bm{a}}\mathM_{N,M}$,
\begin{align}\label{WZ2.86}
  &\big(\p_1,\p_2\big)
  =\frac{1}{2\pi\bm{i}}\oint_{|z|=1}\frac{\left(\frac{\p_1 a(z)}{a'(z)}-
  \frac{\p_1 \hat{a}(z)}{\hat{a}'(z)}\right)\cdot
  \left(\frac{\p_2 a(z)}{a'(z)}- \frac{\p_2 \hat{a}(z)}{\hat{a}'(z)}\right)}
  {\frac{a(z)}{a'(z)}- \frac{\hat{a}(z)}{\hat{a}'(z)}}\frac{dz}{z^2}.
\end{align}
\end{prop}

\subsection{Reduction to finite-dimensional Frobenius manifolds}
At the end of this section, we want to discuss the connection
between the infinite-dimensional Frobenius manifolds
$\mathcal{M}_{N,M}$ and Frobenius manifolds of finite dimension.

Recall that $\mathcal{M}_{N,M}$ is parameterized by two functions
$\zeta(z)$ and $l(z)$. If $\zeta(z)$ is fixed (hence all $t^i$
fixed), then we get a submanifold of $\cM_{N,M}$, which is
$(M+N)$-dimensional and has coordinate $\bm{h}\cup\hat{\bm{h}}$. For
instance, given an arbitrary constant $\ep\ne0$, we let
$\zeta(z)=\ep z$ and have the following submanifold
\begin{equation}\label{}
\mathcal{N}_{N,M}^\ep=\left\{(a(z),\hat{a}(z))=(l(z),l(z)-\ep z) \mid
l(z)=z^N+\sum_{k=-M}^{N-1} v_k z^k\right\}\subset\mathcal{M}_{N,M}.
\end{equation}
Here for simplifying notations we write $\mathcal{N}^\ep$ instead of
 $\mathcal{N}^\ep_{N,M}$. Clearly $\mathcal{N}^\ep$ is a flat submanifold with a metric
induced by the metric \eqref{WZ2.25} on $\cM_{N,M}$. It follows from
\eqref{FMN} that
\begin{equation}\label{}
\mathcal{F}_{N,M}|_{\mathcal{N}^\ep}=F_{N,M}-\ep\,\res_{z=0}l(z).
\end{equation}
However, for general $N$ and $M$ the induced $3$-point correlator
functions do not satisfy the WDVV associativity equations
\eqref{WDVV1} (errors come from the first term on the right hand
side of \eqref{hhat3}). Hence $\mathcal{N}^\ep$ is not a Frobenius
submanifold in Strachan's sense \cite{St}.

Let $\ep\to0$, which means that the curve $\Gamma=\zeta(z)|_{S^1}$
shrinks to a point, then
\begin{equation}\label{}
\mathcal{F}_{N,M}|_{\mathcal{N}^\ep}\to F_{N,M},
\end{equation}
and we obtain a Frobenius manifold structure $\mathcal{N}^0$ that is
given by the potential $F_{N,M}$ together with the unity vector
field $e=\p/\p\hat{h}^M$ and Euler vector field
\[
E=\sum_{j=1}^{N-1}\frac{j}{N}h^j\frac{\p}{\p h^j}+
     \sum_{k=1}^M\frac{k}{M}\hat{h}^k\frac{\p}{\p\hat{h}^k}
     +\frac{N+M}{N}\frac{\p}{\p\hat{h}^0}.
\]
Namely, we recover the structure of Frobenius manifold
$M(\tilde{A}_{M+N-1};N)$ on the orbit space of the extended affine
Weyl group $\widetilde{W}^{(N)}(A_{M+N-1})$, see \cite{DZ2}.
Strictly speaking, $\mathcal{N}^0$ does not belong to $\cM_{N,M}$
but to certain compaction of it (we remark that the abused notion
``submanifold'' in \cite{WX2} can also be understood in this way).

In fact, the above approach that reduces the Frobenius manifold
$\cM_{N,M}$ to $\mathcal{N}^0$ is natural from the viewpoint of
Hamiltonian structures. As to be seen in the following section, on
$\cM_{N,M}$ it is associated with the bi-hamiltonian structure for
the dispersionless Toda lattice hierarchy, which is reduced to the
bi-hamiltonian structure (associated to $M(\tilde{A}_{M+N-1};N)$ )
for the dispersionless extended bigraded Toda hierarchy under the
constraint $a(z)=\hat{a}(z)$, see \cite{Wu} for details.

\section{Relation to dispersionless Toda lattice hierarchy}

Let us study how the infinite-dimensional Frobenius manifolds
constructed above are related to the Toda lattice hierarchy.

In this section we work on the loop space $\mathL\cM_{N,M}$
of smooth maps from $S^1$ to $\cM_{N,M}$. More precisely, the space
$\mathL\cM_{N,M}$ consists of points $\bm{a}=(a(z,x),\hat{a}(z,x))$
of the form \eqref{WZ2.2} with coefficients $v_i$ and $\hat{v}_j$
being smooth functions of $x\in S^1$ constrained by the conditions
(M1)--(M3). The tangent and the cotangent spaces of
$\mathL\cM_{N,M}$ are identified with spaces of Laurent series in a
natural way as \eqref{TaM} and \eqref{WZ2.5}, respectively. Thereby
the pairing between a covector
$\bm{\om}=(\om(z,x),\hat{\om}(z,x))\in T^*_{\bm{a}}\mathL\cM_{N,M}$
and a vector $\bm{X}=(X(z,x),\hat{X}(z,x))\in
T_{\bm{a}}\mathL\cM_{N,M}$ reads (cf. \eqref{WZ2.6})
\begin{equation}\label{}
\la {\bm{\om},\bf X}\ra=\frac{1}{2\pi \bm{i}}\oint_{x\in
S^1}\oint_{|z|=1}\big[\om(z,x)X(z,x)
+\hat{\om}(z,x)\hat{X}(z,x)\big]\dfrac{d z}{z}d x.
\end{equation}
To avoid lengthy notations, we will write $\al(z)$ instead of
$\al(z,x)$ whenever no confusion would happen.

Let
\begin{align}\label{rela}
&\lambda(z)=a(z)^{1/N}=z+\frac{1}{N}v_{N-1}+O(z^{-1}), \quad
|z|\to\infty; \\
& \hat{\lambda}(z)=\hat{a}(z)^{{1}/{M}}=\hat{v}_{-M}^{1/M}z^{-1}
+\frac{\hat{v}_{-M+1}}{M\hat{v}_{-M}^{(M-1)/M}}+O(z), \quad |z|\to0.
\end{align}
A hierarchy of evolutionary equations on the loop space $\mathL
\mathM_{N,M}$ is defined as follows:
\begin{align}
  &\frac{\p a(z)}{\p s_n}=\{(\lambda(z)^n)_{\ge0},a(z)\},\quad
  \frac{\p\hat{a}(z)}{\p s_n}=\{(\lambda(z)^n)_{\ge0},\hat{a}(z)\},\label{disphir01}\\
  &\frac{\p a(z)}{\p
  \hat{s}_n}=\{-(\hat{\lambda}(z)^n)_{<0},a(z)\},\quad
  \frac{\p\hat{a}(z)}{\p\hat{s}_n}=\{-(\hat{\lambda}(z)^n)_{<0},\hat{a}(z)\},
  \label{disphir02}
\end{align}
where $n=1,2,3,\dots$ and the Lie bracket reads
\begin{equation}\label{lieb}
\big\{f,g\big\}:=z\left(\frac{\p f}{\p z} \frac{\p g}{\p x}-\frac{\p
g}{\p z} \frac{\p f}{\p x}\right).
\end{equation}
This hierarchy is the dispersionless limit of the Toda lattice
hierarchy \cite{UT, TT}.

In fact, equations \eqref{disphir01}--\eqref{disphir02} can be
defined equivalently with $(a(z),\hat{a}(z))$ replaced by
$(\lambda(z),\hat{\lambda}(z))$. Based on this fact, denote
$\hat{v}_{-M}^{1/M}=e^u$, then equations
\eqref{disphir01}--\eqref{disphir02} yields
\[
\p_{s_1} e^u=\frac{1}{N}e^u \p_x v_{N-1}, \quad
\frac{1}{N}\p_{\hat{s}_1} v_{N-1}=\p_x e^u.
\]
Eliminating $v_{N-1}$, one has
\begin{equation}\label{}
\p_{s_1}\p_{\hat{s}_1}u=\p_x^2 e^u,
\end{equation}
which is rewritten to \eqref{utxy} with $t=s_1+\hat{s}_1$ and
$y=s_1-\hat{s}_1$.

Consider the following space of local functionals on
$\mathL\cM_{N,M}$:
\[
\mathscr{F}=\left\{F=\int_{S^1}
f(v_{N-1}(x),v_{N-2}(x),\dots,\hat{v}_{-M}(x),\hat{v}_{-M+1}(x),\dots)\,d
x\right\}.
\]
For any $F\in\mathscr{F}$, its variational gradient $d F$ is a
covector field on $\mathL\cM_{N,M}$ such that $\delta F=\la d F,
\delta\bm{a}\ra$.

As derived by one of the authors in \cite{Wu}, the loop space
$\mathL\cM_{N,M}$ is equipped with two compatible Poisson brackets
\begin{equation}\label{poiss0}
 \big\{F,H\big\}_\nu=\la d F,\mathcal{P}_\nu(d H)\ra, \quad
 \nu=1,2
\end{equation}
for any $F,H\in\mathscr{F}$, where $\mathcal{P}_\nu:
T^*_{\bm{a}}\mathL\cM_{N,M}\to T_{\bm{a}}\mathL\cM_{N,M} $ are
Hamiltonian operators. More exactly, for any covector
$\bm{\om}=(\om(z,x),\hat{\om}(z,x))\in T^*_{\bm{a}}\mathL\cM_{N,M}$
we have
\begin{align}
  \mathcal{P}_1&(\om(z),\hat{\om}(z)) \nn\\
  =&\Big(-
  \{(\om(z)+\hat{\om}(z))_{<0}, a(z)\}+(\{\om(z),a(z)\}+\{\hat{\om}(z),\hat{a}(z)\})_{\le0}, \nn
\\
  &\quad \{(\om(z)+\hat{\om}(z))_{\ge0},\hat{a}(z)\}-
  (\{\om(z),a(z)\}+\{\hat{\om}(z),\hat{a}(z)\})_{>0}
  \Big),
  \label{poiss01}\\
  \mathcal{P}_2&(\om(z),\hat{\om}(z))\nn\\
  =&\Big(-
  \{(a(z)\om(z)+\hat{a}(z)\hat{\om}(z))_{<0}, a(z)\}
  +a(z)(\{\om(z),a(z)\}+\{\hat{\om}(z),\hat{a}(z)\})_{\le0}, \nn
\\
  &\quad \{(a(z)\om(z)+\hat{a}(z)\hat{\om}(z))_{\ge0},\hat{a}(z)\}-
  \hat{a}(z)(\{\om(z),a(z)\}+\{\hat{\om}(z),\hat{a}(z)\})_{>0}
  \Big) \nn\\
  &+(z a'(z) \p_x f, z\hat{a}'(z)\p_x f)\label{poiss02}
  \end{align}
with
\[
f=\frac{1}{N}\frac{1}{2\pi\bm{i}}\oint_{|z|=1}\big(\om(z)a'(z)
+\hat{\om}(z)\hat{a}'(z)\big)d z.
\]
The fact that $\mathcal{P}_\nu(\om(z),\hat{\om}(z))$ takes value in
$z^{N-1}\mathH^-\times z^{-M}\mathH^+$ (see \eqref{TaM}) follows
from the form of $(a(z),\hat{a}(z))$ and of the Lie bracket
\eqref{lieb}; the only doubt might be the first component of
\eqref{poiss02}, in which the coefficient of $z^N$ is
\begin{align*}
&\frac{1}{2\pi\bm{i}}\oint_{|z|=1}\big(\{\om(z),a(z)\}+\{\hat{\om}(z),\hat{a}(z)\}\big)\frac{d
z}{z} \\
&\quad
+N\cdot\frac{1}{N}\frac{1}{2\pi\bm{i}}\p_x\oint_{|z|=1}\big(\om(z)a'(z)
+\hat{\om}(z)\hat{a}'(z)\big)d z \\
=&\frac{1}{2\pi\bm{i}}\oint_{|z|=1}\big(\om'(z)\p_x a(z)+\om(z)\p_x
a'(z)+\hat{\om}'(z)\p_x\hat{a}(z)+ \hat{\om}(z)\p_x\hat{a}'(z)\big)d z \\
=& \frac{1}{2\pi\bm{i}}\oint_{|z|=1}\p_z\big(\om(z)\p_x
a(z)+\hat{\om}(z)\p_x\hat{a}(z)\big)d z=0.
\end{align*}
One observes the slight difference between these Poisson structures
and those given in \cite{Ca-rm} and \cite{CDM} for $M=N=1$.

\begin{prop}[\cite{Wu}]
The dispersionless Toda lattice hierarchy
\eqref{disphir01}--\eqref{disphir02} can be represented to a
bi-hamiltonian form as
\begin{align}\label{Hamtd}
\frac{\p F}{\p s_n}=\{F, H_{n+N}\}_1=\{F, H_{n}\}_2, \quad \frac{\p
F}{\p \hat{s}_n}=\{F, \hat{H}_{n+M}\}_1=\{F, \hat{H}_{n}\}_2
\end{align}
with  $n=1,2,3,\dots$ and
\begin{align}
  H_n=\frac{N}{n}\frac{1}{2\pi\bm{i}}\oint_{S^1}\oint_{|z|=1}\lambda(z)^{n}\frac{d z}{z}\,d x,\
  \hat{H}_n=\frac{M}{n}\frac{1}{2\pi\bm{i}}\oint_{S^1}\oint_{|z|=1}\hat{\lambda}(z)^{n}\frac{d z}{z}\,d x.
\end{align}
\end{prop}

\begin{prop}
The bi-hamiltonian structure \eqref{poiss01}--\eqref{poiss02}
coincides with the one induced by the pencil consisting of the
metric \eqref{WZ2.14} and the intersection form \eqref{WZ2.81} on
$\cM_{N,M}$.
\end{prop}

\begin{proof}
We introduce two generating functions for local functionals
\begin{equation}
 a(p,y)=p^{N}+\sum_{i\leq N-1}v_i(y) p^{i}, \quad
   \hat{a}(p,y)=\sum_{j\geq -M}\hat{v}_j(y) p^{j}.
\end{equation}
Their variational gradients are (cf.
\eqref{rpgener01}--\eqref{rpgener02}) respectively
  \begin{align}\label{oneform}
  & da(p,y)=\left(\frac{p^{N}}{z^{N-1}(p-z)}\delta(x-y),0\right),\quad
  |z|<|p|;
  \\
   & d\hat{a}(p,y)=\left(0,\frac{z^{M+1}}{p^{M}(z-p)}\delta(x-y)\right),\quad
  |z|>|p|.
  \end{align}
Substituting them into \eqref{poiss0}, then by a straightforward
calculation we have
 \begin{align}
 \label{bh1}
  \{\al(p,x),\beta(q,y)\}_1=&\frac{p q}{p-q}\big(\al'(p)-\beta'(q)\big)\delta'(x-y)\nn\\
  &+p q\left(\frac{\p_x\al(p)-\p_x\beta(q)}{(p-q)^2}-\frac{\p_x\beta'(q)}{p-q}\right)\delta(x-y),
\\
  \{\al(p,x),\beta(q,y)\}_2=&p q\left(\frac{\al'(p)\beta(q)-\al(p)\beta'(q)}{p-q}
  +\frac{\al'(p)\beta'(q)}{N}\right)\delta'(x-y)\nn\\
  &+p q\bigg(\frac{\p_x\al(p)\cdot\beta(q)-\al(p)\p_x\beta(q)}{(p-q)^2}
  \nn\\
  &+\frac{\al'(p)\p_x\beta(q)-\al(p)\p_x\beta'(q)}{p-q}+\frac{\al'(p)
  \p_x\beta'(q)}{N}\bigg)\delta(x-y),
  \label{bh2}
  \end{align}
where $\al,\beta\in\{a,\hat{a}\}$. These Poisson brackets are of
hydrodynamic type \cite{DN}.

Observe that the coefficients of $\delta'(x-y)$ in
\eqref{bh1}--\eqref{bh2} are exactly the same with the generating
functions for the metrics \eqref{WZ2.14} and \eqref{WZ2.81} on the
Frobenius manifold $\cM_{N,M}$. Therefore, according to the theory
of \cite{DN, Du}, the proposition is proved.
\end{proof}

One can see that the densities of Hamiltonian functionals $H_n$
($n=1,2,\dots,N$) and $\hat{H}_m$ ($m=1,2,\dots,M$) are multiple or
linear combination of flat coordinates in
\[
\{h^j\mid  j=1,\dots,N-1\}\cup\{\hat{h}^k\mid k=1,\dots,
M\}\cup\{t^{i}\}.
\]
By virtue of the bi-hamiltonian recursion relations \eqref{Hamtd},
we conclude that the dispersionless Toda lattice hierarchy
\eqref{disphir01}--\eqref{disphir02} is a subhierarchy of the
principal hierarchy associated to the infinite-dimensional Frobenius
manifold $\cM_{N,M}$.

In order to write down the whole principal hierarchy for
$\cM_{N,M}$, one needs to find flat coordinates for the deformed
flat metric on this infinite-dimensional Frobenius manifold, which
is still open in general. We remark that Carlet and Mertens
\cite{CM} obtained the principal hierarchy for the
infinite-dimensional Frobenius manifold $M_0$ (cf. $\cM_{1,1}$)
constructed in \cite{CDM}.

\begin{rem}
In \cite{BMRWZ} Boyarsky et al. derived a system of WDVV
associativity equations satisfied by the logarithm of tau function
of the dispersionless Toda lattice hierarchy (see also \cite{BX1,
BX2}). It is interesting to compare solutions to their associativity
equations and the potential $\mathcal{F}_{N,M}$ given in the present
paper. This might be clarified from an answer to a more general
question, that is, as indicated at the end of \cite{CDM}, how to
extend the potential $\mathcal{F}_{N,M}$ to the so called
topological solution of the principal hierarchy associated to
$\cM_{N,M}$. We will study it elsewhere.
\end{rem}

\section{Conclusions and remarks}

We have obtained a class of infinite-dimensional Frobenius manifolds
for the dispersionless Toda lattice hierarchy, which generalizes the
construction of \cite{CDM}. Moreover, from these
infinite-dimensional manifolds we rederive the structure of
Frobenius manifolds that were constructed on the orbit space of
extended affine Weyl groups of type A in \cite{DZ2}. This
relationship is consistent with the relation between the Toda
lattice and the extended bigraded Toda \cite{Ca-bt, UT} hierarchies
as well as their bi-hamiltonian structures \cite{Wu}.

In the particular case $N=M=1$, our Frobenius manifold $\cM_{1,1}$
is not exactly the same with $M_0$ given in \cite{CDM}. by a slight
difference (see \eqref{WZ2.14} and \eqref{metricCDM}). Up to now we
do not see any direct connection between these two Frobenius
manifolds, which might be hinted by the intertwining operators in
\cite{Wu} that relate the $r$-matrices on a Lie algebra for
constructing Hamiltonian structures for the Toda lattice hierarchy.
Such Hamiltonian structures deduced from different $r$-matrices
probably imply some other infinite-dimensional Frobenius manifolds
besides $\cM_{1,1}$ and $M_0$. We plan to consider these questions
in the future.

The analogy between the Frobenius manifolds underlying the Toda
lattice and the two-component BKP hierarchies is clear. In fact, the
potential $\mathcal{F}_{m,n}$ of the infinite-dimensional Frobenius
manifold in \cite{WX2} labeled by positive integers $m$ and $n$ can
be written more rigorously as
\begin{align*}
     \mathcal{F}_{m,n}=&\left(\frac{1}{2\pi\bm{i}}\right)^2
     \oint\oint_{|z_1|<|z_2|}\left(\frac{1}{2}\zeta(z_1)\zeta(z_2)-\zeta(z_1)l(z_2)
     +l(z_1)\zeta(z_2)\right)\times \nn \\
     & \qquad\qquad\qquad
     \times\log\left(\frac{z_2-z_1}{z_2}\right)dz_1dz_2+F_{m,n},
\end{align*}
where $\zeta(z)$ and $l(z)$ are certain even functions similar as
\eqref{WZ2.12}, and $F_{m,n}$ is the potential of the Frobenius
manifold defined on the orbit space of the Coxeter group of type
$B_{m+n}$, see \cite{Ber3, Du, Z2007} for details. Note that the
replacement of the factor $\log(z_2-z_1)$ in \cite{WX2} (see (1.19))
by $\log\left(\frac{z_2-z_1}{z_2}\right)$ does not change the result
due to the residue of an even function always vanishes. It shows
that the approach in \cite{CDM, WX2} and the present paper is indeed
efficient for constructing infinite-dimensional Frobenius manifold
related to two-component difference or differential integrable
hierarchies. What is more, in such constructions,
infinite-dimensional Frobenius manifolds can be naturally connected
to Frobenius manifolds of finite dimension. We hope that this
observation provides a hint to enlarge the family of
infinite-dimensional Frobenius manifolds, as well as to understand
their relation with the Frobenius manifolds given in \cite{Ra, Sz}.

\noindent{\bf Acknowledgments.} The authors are grateful to
Professors Boris Dubrovin, Youjin Zhang and Qing Chen
for constant supports and helpful comments. They also thank the
editors and the referees for their patience and constructive suggestions,
as well as referring the authors to the paper \cite{St}. D.\,Zuo also
thanks the University of Glasgow for the hospitality  and Professor
Ian Strachan for elaborately explaining his paper \cite{St}.
The research of
C.-Z.\,Wu has received specific funding under the ``Young SISSA
Scientists' Research Projects'' scheme 2012-1013, promoted by the
International School for Advanced Studies (SISSA), Trieste, Italy.
The research of D.\,Zuo is partially supported by NSFC (11271345, 11371138),
NCET-13-0550, ``PCSIRT", SRF for
ROCS,SEM and OATF,USTC.


\end{document}